\documentclass[a4paper, 10pt, onecolumn] {article}
\pdfoutput=1
\usepackage[T1]{fontenc}
\usepackage[total={6.5in,8.75in},top=1.in, left=0.9in,bottom=1.in, includefoot]{geometry}
\usepackage[utf8]{inputenc}
\usepackage[english]{babel}
\usepackage{amsmath}
\usepackage{amsfonts}
\usepackage{amssymb}
\usepackage{xcolor}
\usepackage{mathrsfs}
\usepackage{amsthm}
\usepackage{float}
\usepackage{authblk}
\usepackage{comment}
\usepackage{graphicx}        
\usepackage[parfill]{parskip}
\usepackage{prettyref}

\numberwithin{equation}{section}

\long\def \beq#1\eeq {\begin{equation} #1 \end{equation}}
\long\def \beaq#1\eeaq {\begin{equation}\begin{aligned} #1 \end{aligned}\end{equation}}
\long\def \bes#1\ees {\begin{equation}\begin{split} #1 \end{split} \end{equation}}
\long\def \bea#1\eea {\begin{eqnarray} #1 \end{eqnarray}}
\long\def \bse[#1]#2\ese {\begin{subequations}\label{#1}\begin{align} #2 \end{align}\end{subequations}}

\long\def\dm[#1]{\!\operatorname{d\mu}\left(#1\right)}

\newtheorem{innercustomgeneric}{\customgenericname}
\providecommand{\customgenericname}{}
\newcommand{\newcustomtheorem}[2]{%
  \newenvironment{#1}[1]
  {%
   \renewcommand\customgenericname{#2}%
   \renewcommand\theinnercustomgeneric{##1}%
   \innercustomgeneric
  }
  {\endinnercustomgeneric}
}

\theoremstyle{plain}
\newtheorem{Remark}{Remark}
\newtheorem{Theorem}{Theorem}
\newtheorem{Lemma}{Lemma}
\newtheorem{Proposition}{Proposition}

\newtheorem{Definition}{Definition}

\newcustomtheorem{customlemma}{Proposition}

\title{The relativistic Hopfield model\\ with correlated patterns}
\author[a]{Elena Agliari}
\author[b]{Alberto Fachechi}
\author[a]{Chiara Marullo}
\date{}

\affil[a]{Dipartimento di Matematica ``Guido Castelnuovo'', Sapienza Universit\`a di Roma, Italy}
\affil[b]{Dipartimento di Matematica e Fisica ``Ennio De Giorgi'', Universit\`a del Salento, Italy}

\begin{document}
\maketitle

\begin{abstract}
In this work we introduce and investigate the properties of the ``relativistic'' Hopfield model endowed with \emph{temporally correlated} patterns. First, we review the ``relativistic'' Hopfield model and we briefly describe the  experimental evidence underlying correlation among patterns. Then, we face the study of the resulting model exploiting statistical-mechanics tools in a low-load regime. More precisely, we prove the existence of the thermodynamic limit of the related free-energy and we derive the self-consistence equations for its order parameters. These equations are solved numerically to get a phase diagram describing the performance of the system as an associative memory as a function of its intrinsic parameters (i.e., the degree of noise and of correlation among patterns). We find that, beyond the standard retrieval and ergodic phases, the relativistic system exhibits \emph{correlated} and \emph{symmetric} regions -- that are genuine effects of temporal correlation -- whose width is, respectively, reduced and increased with respect to the classical case.
\end{abstract}

\section{Introduction}\label{sec:intro}
The Hopfield model is the prototype for neural networks meant for associative memory tasks. Its popularity is mainly due to its similarities with biological neural networks and to the fact that it can be faced analytically exploiting statistical-mechanics tools typical of disordered systems \cite{Amit,Coolen,Bovier}. In the last decade the upsurge of interest in artificial-intelligence applications has determined a renewed interest in neural networks and many variations on theme of the Hopfield model have been proposed, aiming to improve its performance as associative memory and/or to get a picture closer to biology (see e.g., \cite{Agliari-PRL1,Ton,Agliari-PRL2,Albert2}).\\
Among these, the ``relativistic'' Hopfield model \cite{Albert1,Notarnicola,Mechanics} has been introduced based on a formal analogy between the Hopfield model and a mechanical system where the Mattis magnetization (intrinsically bounded) plays as the velocity of a point mass. 
More precisely, the classical Hopfield Hamiltonian, reading as $H^{\textrm{cl}} ( \boldsymbol{ m})= -N  \boldsymbol{m}^2/2$, where $\boldsymbol{m}$  is the Mattis magnetization and $N$ is the system size, can be interpreted as the classical kinetic energy associated to a fictitious particle; its relativistic counterpart can therefore be obtained by applying the transformation
${\boldsymbol{m}}^2/2 \to \sqrt{1 + {\boldsymbol{m}}^2}$, that is, the ``relativistic'' Hopfield Hamiltonian reads as $H^{\textrm{rel}} ( \boldsymbol{ m})= -N \sqrt{1 + \boldsymbol{m}^2}$.
\newline
The relativistic model is as well able to work as associative memory, namely, it can be used to retrieve patterns of information (typically encoded by binary vectors $\boldsymbol \xi$), and its performance turns out to be comparable with the classical counterpart \cite{Albert1}. Interestingly, one can notice that, while the classical model is described only by pairwise interactions, its relativistic version (if expanded in Taylor series with respect to the order parameter $ \boldsymbol{m}$) is an infinite sum of (even) terms that describe higher-order interactions with alternating signs; the terms with negative sign can be related to unlearning mechanisms and play an important role in destabilizing the recall of spurious patterns \cite{Albert2,Albert1,Alem}.

In this work we consider the ``relativistic'' Hopfield model and we allow for patterns exhibiting temporal correlation. 
This refers to well-known experimental facts (see e.g., \cite{Amit2,Miyashita-1988, MiyashitaChang-1988}) where a temporal correlation among patterns during learning (say, $\boldsymbol \xi^{\nu}$ is repeatedly presented after pattern $\boldsymbol \xi^{\mu}$) can, upon simple stimulation, yield to the simultaneous retrieval of correlated patterns (that is, stimulation by pattern $\boldsymbol \xi^{\mu}$ can evoke both $\boldsymbol \xi^{\mu}$ and $\boldsymbol \xi^{\nu}$). 
In order to implement this mechanism into the ``relativistic'' Hopfield model we will follow the route paved by \cite{GriniastyTsodyksAmit-1993,Cugliandolo-1993, CugliandoloTsodyks-1994} for the classical model.

The resulting network, exhibiting relativistic and temporal correlation features, is addressed in the low-load regime exploiting statistical-mechanics tools. 
In particular, we prove the existence of the thermodynamic limit for this model and we get an explicit expression for its thermodynamic pressure by exploiting rigorous methods based on Guerra's interpolating techniques. The extremization of this expression allows us to get self-consistent equations for the order parameter of the model which are then solved numerically. Interestingly, according to the value of the system parameters, i.e., the degree of correlation $a$ and the degree of noise $T$, the solution is qualitatively different. We recall that in the (low-load) classical Hopfield model, large values of $T$ correspond to an ergodic phase, while small values correspond to a retrieval phase, where the system can work as an associative memory, being able to retrieve correctly a certain pattern when this is presented only partially correct. This is still the case for the current system as long as the correlation $a$ is small enough, while when $a$ is relative large a \emph{symmetric} phase (when $T$ is also relatively large) and a \emph{correlated} phase (when $T$ is relatively small) emerge. In particular, in the symmetric case, all patterns are recalled simultaneously and to the same extent; in the correlated phase the stimulating pattern is retrieved only partially and temporally closed patterns are retrieved as well, although to a lower extent.

This paper is structured as follows: in Secs.~\ref{sec:rel} and \ref{sec:tempo} we briefly review the relativistic model and the correlated classical model, respectively; then, in Sec.~\ref{sec:merge}, we merge the two generalizations and define the model on which we shall focus, providing an expression for its intensive free-energy and self-consistency equations for its order parameters; next, in Sec.~\ref{sec:numerics} we get a numerical solution for the behavior of the order parameters, as the system parameters are tuned, highlighting the genuinely relativistic features of the models; finally, Sec.~\ref{sec:conclusions} is left for discussions. Technical details are collected in the Appendices.

\section{The relativistic model} \label{sec:rel}

In this section we briefly recall the main definitions of the ``relativistic'' Hopfield model that we are going to consider and generalize in the current work.  The system is made up of $N$ binary neurons, whose state, either inactive or active, is denoted by $\sigma_i \in \{-1, +1\}$, $i=1, ..., N$. The system is embedded in a complete graph in such a way that each neuron influences and is affected by the state of all the remaining $N-1$ neurons.\\
We also introduce $P$ binary vectors of length $N$, denoted by $\boldsymbol{\xi}^{\mu}$, whose entries are binary, i.e.,  $\xi_i^{\mu} \in \{-1, +1\}$, and extracted i.i.d. with equal probability, for $i=1,...,N$ and $\mu=1,...,P$. These binary vectors are meant as patterns that the neural network is designed to \emph{retrieve}. Before explaining this concept in more details it is convenient to introduce a few more quantities. 
\newline
As anticipated in the previous section, the Hamiltonian of the relativistic model reads as 
\begin{equation}
\label{hamiltoniananrelativistica}
H_{N}^{\textrm{rel}} ( \boldsymbol{ \sigma}| \boldsymbol{\xi}) := -N \sqrt{1+ \boldsymbol{m}_N^2},
\end{equation}
where $\boldsymbol{m}_N$ is the Mattis magnetization, whose $\mu$-th entry is defined as
\begin{equation}
m^{\mu}_N := \frac{1}{N} \sum_{i=1}^N \xi_i^{\mu} \sigma_i,
\end{equation}
namely it measures the alignment of the neural configuration with the $\mu$-th pattern, and it plays as order parameter for the model. Notice that, here and in the following, the subscript $N$ highlights that we are dealing with a finite-size system and it will be omitted when taking the thermodynamic limit.  Exploiting the last definition, and assuming $|\boldsymbol{m}_N | < 1$, the Hamiltonian (\ref{hamiltoniananrelativistica}) can be Taylor-expanded in terms of neuron states and pattern entries as
\begin{equation}
\nonumber
- \frac{H_N^{\textrm{rel}}(\boldsymbol{\sigma}| \boldsymbol{\xi})}{N}= 1 + \frac{1}{2N^2} \sum_{ij} \left(\sum_{\mu=1}^P {\xi}_i  ^\mu{\xi}_j^\mu \right) \sigma_i \sigma_j - \frac{1}{8N^4} \sum_{ijkl} \left(\sum_{\mu=1}^P {\xi}_i ^\mu {\xi}_j^\mu \right) \left(\sum_{\nu=1}^P{\xi}_k  ^\nu{\xi}_l^\nu \right) \sigma_i \sigma_j \sigma_k \sigma_l+ \mathcal{O}(\sigma^6).
\end{equation}
This expression highlights that the ``relativistic'' model includes higher-order interactions among spins.
Next, we introduce the Boltzmann-Gibbs measure for the model described by (\ref{hamiltoniananrelativistica}) as
\begin{equation} \label{eq:gibbs}
\mathcal{G}_{N,\beta}^{\textrm{rel}}(  \boldsymbol \sigma \vert \boldsymbol \xi) := \frac{e^{-\beta H_{N}^{\textrm{rel}} ( \boldsymbol{ \sigma}| \boldsymbol{\xi})}}{Z_{N,\beta}^{\textrm{rel}}(\boldsymbol \xi)},\quad  Z_{N,\beta}^{\textrm{rel}}( \boldsymbol \xi):=\sum_{ \{\boldsymbol \sigma \}}e^{-\beta H_{N}^{\textrm{rel}} ( \boldsymbol{ \sigma}| \boldsymbol{\xi})},
\end{equation}
where $1/\beta \in \mathbb{R}^+$ accounts for the noise level in the system (in such a way that for $\beta \to 0$ the neuronal configuration is completely random, while in the $\beta \to \infty$ limit the Hamiltonian plays as a Lyapounov function) and where $Z_{N,\beta}^{\textrm{rel}}( \boldsymbol \xi )$, referred to as partition function, ensures the normalization of the Gibbs measure.

Focusing on the so-called low-load regime, that is, as the number $N$ of neurons is made larger and larger, the number $P$ of stored patterns grows sub-linearly with $N$, namely
\begin{equation}
\alpha:= \lim_{N \to \infty} \frac{P}{N}=0,
\end{equation}
one finds that, in the thermodynamic limit, the expectation of the Mattis magnetization is given by \cite{Notarnicola}
\begin{equation} \label{eq:sol_r}
 \langle m^{\mu} \rangle_{\textrm{rel}}:= \lim_{N \to \infty} \langle m_N^{\mu} \rangle_{\textrm{rel}} = \mathbb{E} \left[ \xi^{\mu} \tanh \left( \beta \sum_\mu \frac{{\xi}^\mu \langle {m}_\mu \rangle_{\textrm{rel}} }{\sqrt{1+ \langle \boldsymbol{m} \rangle_{\textrm{rel}}^2}} \right)  \right],
\end{equation}
where  the average $\langle \cdot \rangle_{\textrm{rel}}$ is meant with respect to the Gibbs measure (\ref{eq:gibbs}), while the average $\mathbb{E}$ is meant over the pattern realization, namely $\mathbb{E} := 2^{-NP} \prod_{i, \mu=1}^{N,P}  \sum_{ \xi_i^{\mu} = \pm 1}$.
Of course, in the limit $ | \langle \boldsymbol{m}\rangle_{\textrm{rel}} | \ll 1$, we can expand the solution (\ref{eq:sol_r}) 
and recover the classical result \cite{Amit,Coolen}.\footnote{Of course, this expansion is only formal, since the two models are not related by a Taylor expansion. Indeed, also for the classical Hopfield model we have $|\langle \boldsymbol m \rangle_{\text{rel}}|\sim 1$ in the retrieval regime at low thermal noise.}

The self-consistency equation (\ref{eq:sol_r}) can be solved numerically. When $\beta$ is small (i.e. the thermal noise is high), the only solution is given by $\langle m^{\mu} \rangle_{\textrm{rel}} = 0$ $\forall \mu$, which corresponds to an ergodic system unable to retrieve. On the other hand, when $\beta$ is large (and thus the thermal noise is low), the system can relax to final configurations whose Mattis magnetization vector $\langle \boldsymbol{m} \rangle_{\textrm{rel}}$ satisfies $\langle m^{\mu} \rangle_{\textrm{rel}} \neq 0$ and $\langle m^{\nu} \rangle_{\textrm{rel}} = 0, \forall \ \nu \neq \mu$, which is interpreted as the \emph{retrieval} of the $\mu$-th pattern. This is also confirmed by extensive Monte Carlo simulations, as shown in \cite{Albert1}. There, the Authors found that the ``relativistic'' Hopfield model (at least in the low-load regime) is dynamically more sensible to thermal noise w.r.t. its classical counterpart. This is true both for random initial conditions and for starting configurations which are aligned to spurious states. Indeed, the largest amount of thermal noise still ensuring retrieval is lower in the ``relativistic'' Hopfield model and this is a consequence of the fact that in this model energetic wells are shallower w.r.t. to those of the Hopfield network, in such a way that the probability to escape from these wells is higher with respect to the classical reference.\footnote{We stress that this difference is only dynamical. Indeed, the critical temperature for the transition to the ergodicity is fixed to $\beta_c=1$ (in the thermodynamic limit), which is the same as the classical Hopfield model.}

\section{Temporally correlated patterns} \label{sec:tempo}
The pairwise contribution in the Hamiltonian (\ref{hamiltoniananrelativistica}) is given by the standard Hebbian coupling ($J_{ij}^{\textrm{hebb}}:= \frac{1}{N}\sum_{\mu=1}^P \xi_i^{\mu} \xi_j^{\mu}$ for the couple ($i,j$)). The latter can be generalized in order to include more complex combinations among patterns. For instance, we can write
\begin{equation} \label{eq:J_A}
J_{ij} = \frac{1}{N} \sum_{\mu, \nu =1}^{P,P} \xi_i^{\mu} X_{\mu \nu} \xi_j^{\nu},
\end{equation}
where $\boldsymbol{X}$ is a symmetric matrix.  For example, the model by Personnaz {\it et al.} \cite{Personnaz} (later studied by Kanter and Sompolisnky \cite{KanterSompo} from the thermodynamical perspective) and the removal\&consolidation model \cite{Albert2,Alem} belong to this class of networks. 
 Another interesting example is represented by the Hopfield model with minimal\footnote{The adjective {\it minimal} stresses that the temporal correlation only involves closest patterns, i.e. $(\mu,\mu+1)$ and $(\mu-1,\mu)$.} temporal correlation, meaning that patterns are coupled to the nearest ones with an interaction strength $a \in [0, 1]$ (see also previous investigations in \cite{Amit2,GriniastyTsodyksAmit-1993,Cugliandolo-1993,CugliandoloTsodyks-1994,Agliari-Dantoni}). In mathematical terms, the coupling matrix is given by
\begin{equation}\label{eq:JJ_A}
J_{ij}^{\text{corr}}=\frac1N\sum_{\mu=1}^P[\xi^\mu_i \xi^\mu_j+a(\xi^{\mu+1}_i \xi^{\mu}_j+\xi_i^{\mu-1} \xi^\mu _j) ],
\end{equation}
where periodic boundary conditions (i.e., $\xi^{P+1}=\xi^1$ and $\xi^{0}=\xi^P$) are adopted and the compact form \eqref{eq:J_A} is recovered for 
\begin{eqnarray} \label{eq:connection}
\pmb{X}  = \left(
\begin{array}{ccccc}
1 & a &  \cdots & 0 & a  \\
a & 1& \cdots & 0 & 0  \\
\vdots  & \vdots & \ddots & \vdots & \vdots  \\
0 & 0 &  \dots & 1 & a  \\
a & 0 &  \dots & a & 1  \\
\end{array}
\right).
\end{eqnarray}
For the classical Hopfield model the related Hamiltonian therefore reads as
\begin{equation}\label{eq:H_a}
H_{N,a}^{\textrm{cl, corr}}(\boldsymbol \sigma |  \boldsymbol \xi ) = -\frac{1}{2N} \sum_{i,j}^{N,N} \sum_{\mu=1}^P [\xi_i^{\mu} \xi_j^{\mu} + a(\xi_i^{\mu+1} \xi_j^{\mu} +\xi_i^{\mu-1} \xi_j^{\mu} ) ] \sigma_i \sigma_j.
\end{equation}
The corresponding Gibbs measure is
\begin{equation}
\mathcal{G}_{N, \beta, a}^{\textrm{cl, corr}}(  \boldsymbol \sigma\vert \boldsymbol \xi) := \frac{e^{-\beta H_{N,a}^{\textrm{cl, corr}} ( \boldsymbol{ \sigma}| \boldsymbol{\xi})}}{Z_{N,\beta, a}^{\textrm{cl, corr}}( \boldsymbol \xi)},
\end{equation} 
where, again, the partition function $Z_{N,\beta,a}^{\textrm{cl, corr}}(\boldsymbol \xi)$ ensures normalization, and the related Gibbs average is denoted as $\langle  \cdot \rangle_{\textrm{cl, corr}}$.

As anticipated in Sec.~\ref{sec:intro}, this modification of the Hopfield model captures some basic experimental facts \cite{Miyashita-1988, MiyashitaChang-1988}: a temporal correlation among visual stimuli during learning can spontaneously emerge also during retrieval.
%
%
Indeed, the model (\ref{eq:H_a}) is able to reproduce this experimental feature in both low \cite{GriniastyTsodyksAmit-1993,Cugliandolo-1993} and high \cite{CugliandoloTsodyks-1994} storage regimes. For the former, in the thermodynamic limit, the following self-consistent equation for the order parameter holds \cite{GriniastyTsodyksAmit-1993}
\begin{equation} \label{eq:selfcons_a}
\langle m^{\mu} \rangle_{\textrm{cl, corr}} = \mathbb{E} \bigg[ \xi^{\mu} \, \tanh \Big( \beta \sum_{\mu=1}^P \langle m^{\mu} \rangle_{\textrm{cl, corr}} [\xi_i^{\mu} + a (\xi_i^{\mu+1} + \xi_i^{\mu-1}) ] \Big) \bigg ].
\end{equation}

In \cite{GriniastyTsodyksAmit-1993}, the previous equation was solved by starting from a pure pattern state (say, $\boldsymbol \sigma = \boldsymbol \xi^1$) and iterating until convergence.
%
%
In the noiseless case ($\beta \to \infty$), where the hyperbolic tangent can be replaced by the sign function, the pure state is still a solution if $a \in [0, 1/2)$, while if $a \in (1/2,1]$, the solution is characterized by the Mattis magnetizations (assuming $P \geq 10$)
\begin{equation} \label{eq:ansatz_leti}
\pmb{m}^T= \frac{1}{2^7} (77,51,13,3,1,0,...,0,...,0,1,3,13,51),
\end{equation}
namely, the overlap with the pattern $\boldsymbol \xi^1$ used as stimulus is the largest and the overlap with the neighboring patterns in the stored sequence decays symmetrically until vanishing at a distance of $5$. 

In the presence of noise, one can distinguish four different regimes according to the value of the parameters $a$ and $\beta$.
The overall behavior of the system is summarized in the plot of Fig.~\ref{phasediagram} (left panel).
A similar phase diagram, as a function of $\alpha$ and $a$, was drawn in \cite{CugliandoloTsodyks-1994} for the high-storage regime.

\begin{figure}[h!]
\begin{center}
\includegraphics[width=0.9\textwidth]{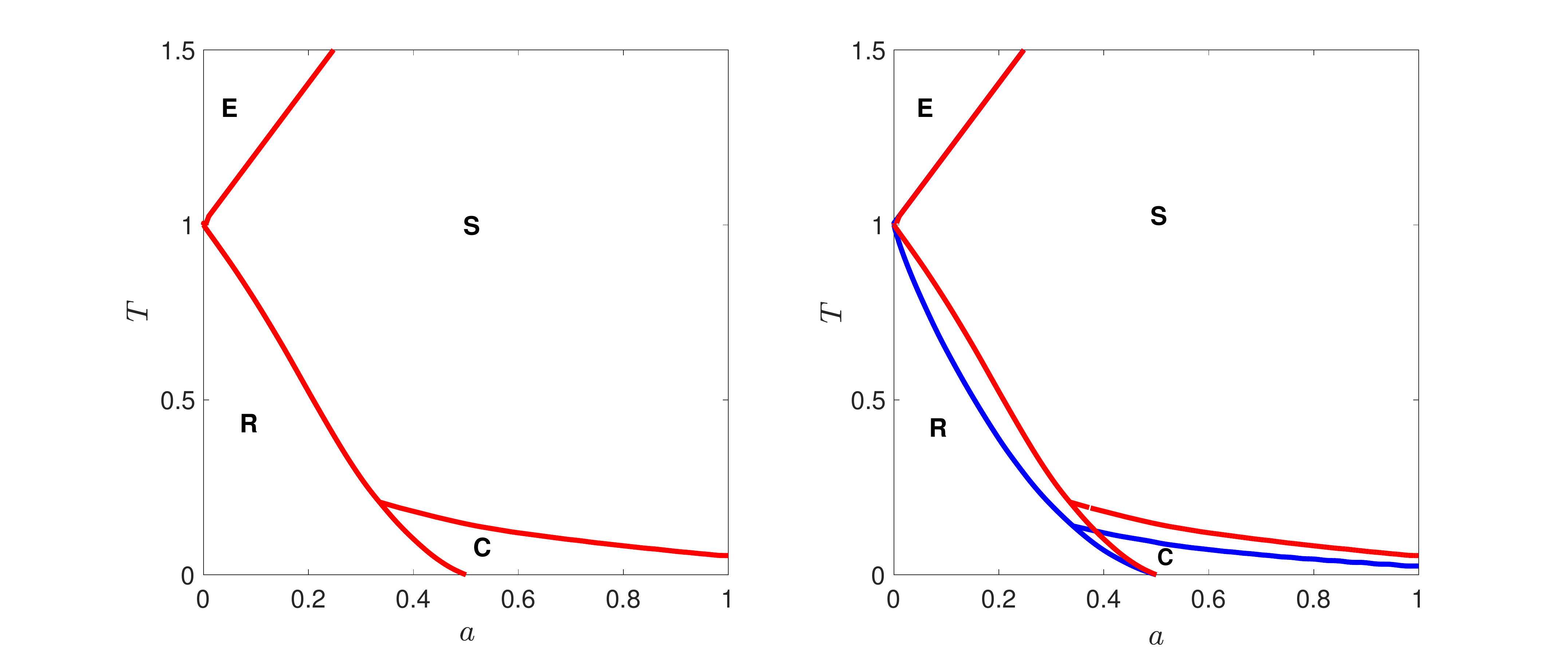}
\caption{Left panel: Phase diagram for the classical correlated model with low storage ($P=5$) described by the Hamiltonian (\ref{eq:H_a}). At a high level of noise the system is ergodic (E) and, for any initial configuration, it eventually reaches a state with $m^{\mu}=0, \forall \mu$. At lower temperatures (below the transition line), the system evolves to a so-called symmetric state (S), characterized by, approximately, $m^{\mu}=m \neq 0, \forall \mu$.
Then, if $a$ is small enough, by further reducing the temperature (below the transition line), the pure state retrieval (R) can be recovered.
On the other hand, if $a$ is larger, as the temperature is reduced, correlated attractors (C) appear according to Eq.~(\ref{eq:ansatz_leti}). Then, if the temperature is further lowered, the system recovers the retrieval state, yet if $a>1/2$, this state is no longer achievable. Right panel: Phase diagram for the relativistic correlated model (solid line) with low storage ($P=5$) described by the Hamiltonian (\ref{Hopfield-rel}) compared with that obtained for the classical case (dashed line) described by the Hamiltonian (\ref{eq:H_a}). Analogous regions (E, R, S, C) emerge, but here the symmetric region is wider having partially invaded the correlated and the retrieval regions, on the other hamd the ergodic phase is unchanged.}
\label{phasediagram}
\end{center}
\end{figure}

\section{The ``relativistic'' Hopfield model with temporally correlated patterns} \label{sec:merge}
The discussion in the previous section only concerns the classical Hopfield model and it is then natural to question about the consequences of temporal correlation between patterns in its relativistic extension. Then, in this Section we consider a neural network which merges the two features described in Sec.~\ref{sec:rel} and  Sec.~\ref{sec:tempo}, respectively. \medskip
\begin{Definition}
Given the temporal correlation strength $a \in [0,1]$, the Hamiltonian of the ``relativistic'' Hopfield model with (minimal) cyclic temporal correlation between patterns is
\begin{equation}
\label{Hopfield-rel}
H_{N,a}^{\textrm{rel,corr}}(\boldsymbol \sigma| \boldsymbol \xi) = - N \sqrt{1+ \frac{1}{N^2}\sum_{\mu=1}^{P}\sum_{i,j=1}^{N,N}{\sigma_i \sigma_j[\xi_i^{\mu}
\xi_j^{\mu} + a(\xi_i^{\mu+1}\xi_j^{\mu}+ \xi_i^{\mu-1}\xi_j^{\mu})}]},
\end{equation}
where $\sigma_i=\pm 1$ $i\in \{1,...,N\}$ are the binary variables representing the neural activities and the entries of the $P$ digital patterns $\boldsymbol{\xi}^\mu$, $\mu \in \{1,...,P\}$ are independently drawn with equal probability $\mathbb{P}(\xi^\mu_i=+1)=\mathbb{P}(\xi^\mu_i=-1)=\frac{1}{2}.$
\end{Definition}
Notice that the Hamiltonian function \eqref{Hopfield-rel} can be put in a more compact form in terms of the Mattis magnetization $\boldsymbol m$ and the correlation matrix $\boldsymbol X$ as
\begin{equation} \label{Hopfield-rel1}
H_{N,a}^{\textrm{rel},\textrm{corr}}(\boldsymbol \sigma| \boldsymbol \xi) := - N \sqrt{1+ \boldsymbol{m}^T \boldsymbol{X}\boldsymbol{m}}.
\end{equation}
This Hamiltonian yields to the Gibbs measure
\begin{equation}
\mathcal{G}_{N,\beta,a}^{\textrm{rel},\textrm{corr}}(\boldsymbol \sigma\vert\boldsymbol \xi) := \frac{e^{-\beta H_{N,a}^{\textrm{rel},\textrm{corr}} ( \boldsymbol{ \sigma}| \boldsymbol{\xi})}}{Z_{N,\beta,a}^{\textrm{rel},\textrm{corr}}(\boldsymbol \xi)},
\end{equation} 
where, as usual, $Z_{N,\beta,a}^{\textrm{rel},\textrm{corr}}( \boldsymbol \xi)$ ensures normalization, and the related Gibbs average is denoted as $\langle  \cdot \rangle_{\textrm{rel},\textrm{corr}}$.
In the following we will drop the superscript and subscript ``$\textrm{rel},\textrm{corr}$'' in order to lighten the notation.

We are now investigating the model from a statistical-mechanics perspective, where the key quantity to look at is the intensive pressure \footnote{We recall that the free energy $\tilde{F}$ equals the pressure $F$, a constant apart, that is $F = -\beta \tilde F$.}.
\begin{Definition}
Using $\beta\in \mathbb R^+$ as the parameter tuning the thermal noise, the intensive pressure associated to the ``relativistic'' Hopfield model with (minimal) cyclic temporal correlation is given by
\begin{equation}\label{original_model}
    F_{N,\beta,a}( \boldsymbol \xi) := \frac{1}{N}  \left [\log Z_{N ,\beta,a }(\boldsymbol \xi) \right] =  \frac{1}{N}  \left[ \log \sum_{ \boldsymbol \sigma } e^{-\beta H_{N,a}(\boldsymbol \sigma| \boldsymbol \xi)} \right].
\end{equation}
\end{Definition}
An important feature of the intensive pressure in the thermodynamic limit is presented in the following\medskip
\begin{Proposition}\label{prop1}
In the thermodynamic limit, the self-average property of the intensive pressure holds, i.e.
\begin{equation}
\lim_{N \to \infty} \mathbb{E} \left \{  F_{N,\beta,a}(\boldsymbol \xi) - \mathbb{E} [F_{N, \beta,a}( \boldsymbol \xi)]  \right \}^2 =0,
\end{equation}
\end{Proposition}
This is quite an expected result for the model under investigation (see e.g., \cite{Guerra2}), and we provide a complete proof in Appendix \ref{sec:automedia}; due to this property, in the thermodynamic limit we can drop the dependence of the intensive pressure on the set of stored patterns, i.e.
\begin{equation}
F_{\beta, a} := \lim_{N \to \infty} F_{N, \beta, a}(\boldsymbol \xi).
\end{equation}
Once the basic objects are introduced, we can turn to the thermodynamical analysis of the ``relativistic'' Hopfield model with correlated patterns. First, we prove the existence of the thermodynamic limit for the intensive pressure, also deriving an explicit expression in terms of the order parameters (the Mattis magnetizations); then, recalling that as $N \to \infty$ the Gibbs measure concentrates around configurations corresponding to maxima of the pressure, we can obtain an estimate of $\langle \boldsymbol m \rangle$ by looking for the extremal points of the pressure.

To follow this route it is convenient to factorize the matrix $\boldsymbol{X}$. Since real symmetric matrices are diagonalizable by orthogonal matrices, $\boldsymbol{X}$ can be written as $ \boldsymbol{X}=\boldsymbol{U}\boldsymbol{D} \boldsymbol{U}^{T}$, where $\boldsymbol{D}$ is a diagonal matrix whose diagonal elements are the eigenvalues of the matrix $\boldsymbol{X}$, while $\boldsymbol{U}$ and $\boldsymbol{U}^T$ are unitary rotation matrices.
Next, we call $ \tilde{\xi}_i^{\mu}:= \sum_{\nu=1}^P (\sqrt{ \boldsymbol{D}} \boldsymbol{U}^{T})_{\mu \nu} ~ \xi_i^{\nu}$, that is, in compact notation,
\begin{equation} 
 \tilde{\boldsymbol{\xi}}:=\sqrt{ \boldsymbol{D}} \boldsymbol{U}^{T}\boldsymbol{\xi}
 \end{equation}
in such a way that the Hamiltonian (\ref{Hopfield-rel}) can be rewritten as
\begin{equation}
     \label{Hopfield-relativistico}
 H_{N,a}(\boldsymbol \sigma| \boldsymbol \xi) = - N \sqrt{1+ \frac{1}{N^2}\sum_{\mu=1}^{P}\sum_{i,j=1}^{N,N}{\sigma_i \sigma_j\tilde{\xi_i^{\mu}}\tilde{
\xi_j^{\mu} }}}.
\end{equation}
Analogously, by introducing the "rotated" Mattis magnetizations $\tilde{{m}}_{N}^\mu= \frac{1}{N} \sum_{i=1}^N \tilde{{\xi}}^\mu_{i}\sigma_{i}$, the Hamiltonian (\ref{Hopfield-relativistico}) can be recast as
\begin{equation} \label{eq:Hruotata}
    H_{N,a}(\boldsymbol \sigma| \boldsymbol \xi)= -N \sqrt{1+ \tilde{\boldsymbol{m}}_N^2}.
\end{equation}
As shown in Appendix \ref{app:Conti1}, the rotation induced by $\boldsymbol X$ preserves pairwise uncorrelation among patterns.

In order to prove the existence of the thermodynamic limit for $F_{N,\beta, a}( \boldsymbol \xi)$ we adopt the scheme originally developed by Guerra and Toninelli \cite{Guerra1,Guerra2} for Hamiltonians that are quadratic forms (and this is the case for the model under study, see eq.~\ref{eq:Hruotata}).
More precisely, omitting the subscript $\beta, a$ to lighten the notation, we will prove that $F_{N}( \boldsymbol \xi)$ is sub-additive with respect to $N$, whence, by Fekete's lemma, $\lim_{N \to \infty} F_{N}( \boldsymbol \xi)$ exists finite and corresponds to the lower bound of the sequence $\{ F_{N}( \boldsymbol \xi) \}$; details are given in the Appendix \ref{sec:esistenza}.

The next step is to obtain an explicit expression for the thermodynamic limit of the intensive pressure $F_{\beta, a}$ in terms of the Mattis magnetizations.
To this aim we introduce the following interpolating pressure with the following \medskip
%
%
\begin{Definition}
 We define an interpolating pressure as
\begin{equation}
\label{alpha-intes}
\bar{F}_N(t) := \frac{1}{N}  \log   \sum_{ \{\boldsymbol \sigma \}}
\exp\Bigg( t \beta N \sqrt{1 + \sum_{\mu=1}^{P} (\tilde{m} ^{\mu})^2} + (1-t) \beta \sum_{\mu=1}^{P} \psi^{\mu} \sum_{i=1}^{N} \tilde{\xi}_i^{\mu}\sigma_i \Bigg),
\end{equation}
where $t \in [0,1]$ is a scalar interpolating parameter and $\psi^{\mu}$, $\mu \in \{1,...,P\}$ are $P$ fields that are functions depending on patterns $\boldsymbol{\tilde{\xi}}^{\mu}$. 
\end{Definition}
At this stage $\{\psi^{\mu}\}_{\mu=1}^P$, can be taken as arbitrary  and their definition will be given a posteriori.
\newline
Notice that the interpolating pressure $\bar F_N(t)$ evaluated at $t=1$ recovers the pressure for the model under study, while, evaluated at $t=0$, recovers the pressure for a system displaying only one-body interactions which can be faced directly. The route we will pursue is to relate these two cases via the fundamental theorem of integral calculus, i.e.
\begin{equation} \label{nuova}
\begin{aligned}
F_{\beta, a} &=\lim_{N \to \infty} \bar{F}_N(t=1) 
=\lim_{N \to \infty}\Big( \bar{F}_N(t=0) + \int_{0}^{1} \frac{d \bar{F}_N(t)}{dt} \bigg \rvert_{t=t'} dt' \Big).
\end{aligned}
\end{equation}
However, before proceeding along this way, a couple of remarks are in order.\medskip
\begin{Remark}
Given a function $f(\boldsymbol \sigma|\boldsymbol {\tilde \xi})$, the interpolating pressure (\ref{alpha-intes}) yields to a generalized average, denoted by $\langle f(\boldsymbol \sigma| \boldsymbol {\tilde \xi} ) \rangle_t$ and defined as
\begin{equation}
\label{Finterpol}
\begin{aligned}
& \langle f(\boldsymbol \sigma| \boldsymbol{ \tilde \xi}) \rangle_t:= 
& \frac{\sum_{\{\boldsymbol \sigma \}} f(\boldsymbol \sigma| \tilde{\boldsymbol{\xi}}) \exp\left( t \beta N \sqrt{1 +
		\tilde{\boldsymbol{m}}_N^2}
 + (1-t) \beta \sum_{\mu=1}^{P} \psi^{\mu} \sum_{i=1}^{N} \tilde{\xi}_i^{\mu}\sigma_i \right) }{\sum_{\{\boldsymbol \sigma\}} \exp\left( t \beta N \sqrt{1 +\tilde{\boldsymbol{m}}_N^2} + (1-t) \beta \sum_{\mu}^{P} \psi^{\mu} \sum_{i=1}^{N} \tilde{\xi}_i^{\mu}\sigma_i \right)}.
\end{aligned}
\end{equation}
\end{Remark}
\begin{Remark}\label{rem:selfaver}
We assume the self-averaging properties of the order parameters, meaning that the fluctuations of the Mattis magnetizations $ \tilde{\boldsymbol m}_N$ with respect to its equilibrium value $\langle \boldsymbol {\tilde m_N}\rangle_t$ vanish in the thermodynamic limit. Indeed, since the network is in the low storage regime, we can expect that it exhibits a ferromagnetic-like behavior. Therefore, it is reasonable to require that the covariance of the magnetizations scales as $N^{-1}$, or equivalently that
\begin{equation}
\label{eq:magnsus}
\lim_{N \to \infty } \big\vert N( \langle  {\tilde m}_N^\mu {\tilde m}_N^\nu\rangle_t -\langle  {{\tilde m}^\mu_N}\rangle_t \langle  {{\tilde m}^\nu_N}\rangle_t)\big\vert < +\infty
\end{equation}
almost everywhere and for all $\mu,\nu=1,\dots,P$. This is a straightforward generalization of the fact that, in ferromagnetic systems, the magnetic susceptibility is finite almost everywhere in the thermodynamic limit. 
\end{Remark}
As a consequence, we can state the following\medskip
\begin{Lemma}\label{prop:constant}
	It is possible to choose the tunable parameters $\boldsymbol \psi$ in order for the generalized-average magnetization $\langle \tilde{\boldsymbol{m}}_N \rangle_t$ (in the thermodynamic limit and under the self-average hypothesis) to be independent on $t$ almost everywhere, i.e.
	$$\frac{d  }{dt}\lim_{N \to \infty } \langle \tilde{\mathbf{m}}_N \rangle_{t}= 0\quad a.e.$$%
\end{Lemma}
\begin{proof}
The constraint \eqref{eq:magnsus} means that the Mattis magnetizations weakly fluctuate around their expectation values for sufficiently large $N$. Thus, we can adopt a formal expression for the Mattis magnetizations as follows:
\begin{equation}
\label{eq:exp}
{\tilde m}^{\mu}_N(\boldsymbol {\sigma})= \langle {\tilde m}_N ^{\mu}\rangle _t+\frac{\Delta_\mu (\boldsymbol{\sigma})}{\sqrt N}.
\end{equation}
Such decomposition is valid almost everywhere. As a direct consequence, we have that the random variables $\Delta_\mu (\boldsymbol{ \sigma})$ have zero mean and finite covariance. We now compute the $t$-derivative of the expectation value of the Mattis magnetization, leaving the possibility for $\psi$ to depend on $t$. By direct calculation we get
\begin{equation}\label{eq:derivative}
\begin{split}
\frac{d \langle \tilde{m}_N^\mu \rangle_{t} }{d t} &	=\beta N\Big\langle\tilde{m}_N^\mu \sqrt{1+ \tilde{\mathbf{m}}_N^2} \Big\rangle_t+\beta N\sum_{\rho=1}^P[-\psi^\rho+(1-t)(\psi^\rho)']\langle\tilde{m}^\rho_N\tilde{m}^\mu_N\rangle_t\\ &-\beta N\langle \tilde{m}_N^\mu \rangle_t\Big\langle\sqrt{1+ \tilde{\mathbf{m}}_N^2} \Big\rangle_t+\beta N\sum_{\rho=1}^P [-\psi^\rho+(1-t)(\psi^\rho)']\langle \tilde{m}_N^\rho \rangle_t\langle\tilde{m}^\mu_N\rangle_t,
\end{split}
\end{equation}
where $(\cdot)'$ stands for the $t$-derivative. In the remainder of the proof, we will drop the subscripts $N$ and $t$ in order to lighten the notation. Applying the expression \eqref{eq:exp}, we can expand the argument of the expectation values in the previous equality around the expectation value of the magnetizations up to the order $N^{-1}$ (because of the prefactor $N$ in \eqref{eq:derivative}). In particular, we have
\begin{equation*}
\begin{split}
\sqrt{1+ \tilde{\mathbf{m}}^2}&= \sqrt{1+\langle \tilde{\mathbf{m}}\rangle^2}+\sum _{\rho=1}^P \frac{\langle\tilde m_\rho \rangle}{\sqrt{1+ \langle\tilde{\mathbf{m}}\rangle^2}}(\tilde m_\rho - \langle \tilde m_\rho \rangle)\\&+\frac12\sum_{\rho,\sigma=1}^P\left(\frac{\delta_{\rho,\sigma}}{\sqrt{1+ \langle\tilde{\mathbf{m}}\rangle^2}}-\frac{\langle  \tilde m_\rho\rangle \langle \tilde m_\sigma\rangle }{({1+ \langle\tilde{\mathbf{m}}\rangle^2})^{3/2}}\right)(\tilde m_\rho - \langle\tilde  m_\rho \rangle)(\tilde m_\sigma - \langle\tilde  m_\sigma \rangle)+ \mathcal R_2(\tilde {\boldsymbol m})=
\\&= \sqrt{1+\langle \tilde{\mathbf{m}}\rangle^2}+\sum _{\rho=1}^P \frac{\langle\tilde  m_\rho \rangle}{\sqrt{1+ \langle\tilde{\mathbf{m}}\rangle^2}}\frac{\Delta_\rho}{\sqrt N}+\frac12\sum_{\rho,\sigma=1}^P\left(\frac{\delta_{\rho,\sigma}}{\sqrt{1+ \langle\tilde{\mathbf{m}}\rangle^2}}-\frac{\langle\tilde  m_\rho\rangle \langle\tilde  m_\sigma\rangle }{({1+ \langle\tilde{\mathbf{m}}\rangle^2})^{3/2}}\right)\frac{\Delta _\rho \Delta_ \sigma}{N}+ \mathcal R_2(\tilde {\boldsymbol m}),
\end{split}
\end{equation*}
where
\begin{eqnarray*}
\mathcal R_2 (\tilde {\boldsymbol m})&=&\frac{1}{3!}\sum_{\rho,\sigma,\eta} \Bigg(\frac{\delta_{\rho,\sigma} H_\eta}{\sqrt{1+\sum_\gamma  H_\gamma ^2 }}-\frac{\delta_{\rho,\eta}  H_\sigma+\delta_{\sigma,\eta} H_\rho}{({1+\sum_\gamma  H_\gamma ^2 })^{3/2}}+3\frac{  H_\rho  H_\sigma  H_\eta}{({1+\sum_\gamma  H_\gamma ^2 })^{5/2}}\Bigg)\frac{\Delta_\rho\Delta_\sigma\Delta_\eta}{N^{3/2}},\\
\boldsymbol H& = &\langle \tilde{\boldsymbol{m}}\rangle+x \boldsymbol \Delta,\quad \text{for some } x \in(0,1),
\end{eqnarray*}
is the Lagrange remainder of order 2. It is easy to prove that
\begin{equation*}
\left\vert\frac{\delta_{\rho,\sigma} H_\eta}{\sqrt{1+\sum_\gamma  H_\gamma ^2 }}-\frac{\delta_{\rho,\eta}  H_\sigma+\delta_{\sigma,\eta} H_\rho}{({1+\sum_\gamma  H_\gamma ^2 })^{3/2}}+3\frac{  H_\rho  H_\sigma  H_\eta}{({1+\sum_\gamma  H_\gamma ^2 })^{5/2}}\right\vert \le 6 \underset{\rho,t\in(0,1)}{\text{max}}\vert H_\rho \vert= 6C,
\end{equation*}
thus
\begin{equation*}
\vert \mathcal R_2 (\tilde {\boldsymbol m}) \vert \le \frac{C}{N^{3/2}} \left(\sum_{\rho=1}^P \vert \Delta_\rho \vert \right)^3,
\end{equation*}
and therefore the remainder scales as $N^{-3/2}$. Using this result, we can compute each expectation value appearing in \eqref{eq:derivative}. In fact,
\begin{equation*}
\begin{split}
\Big\langle\tilde{m}_\mu& \sqrt{1+ \tilde{\mathbf{m}}^2} \Big\rangle=\Big\langle\big( \langle {\tilde m}_{\mu}\rangle +\frac{\Delta_\mu }{\sqrt N}\big) \Big[\sqrt{1+\langle \tilde{\mathbf{m}}\rangle^2}+\sum _{\rho=1}^P \frac{\langle\tilde  m_\rho \rangle}{\sqrt{1+ \langle\tilde{\mathbf{m}}\rangle^2}}\frac{\Delta_\rho}{\sqrt N}\\&+\frac12\sum_{\rho,\sigma=1}^P\left(\frac{\delta_{\rho,\sigma}}{\sqrt{1+ \langle\tilde{\mathbf{m}}\rangle^2}}-\frac{\langle \tilde m_\rho\rangle \langle\tilde  m_\sigma\rangle }{({1+ \langle\tilde{\mathbf{m}}\rangle^2})^{3/2}}\right)\frac{\Delta _\rho \Delta_ \sigma}{N}+ \mathcal R_2(\tilde{\boldsymbol{m}})\Big] \Big\rangle=\\
&=
\langle \tilde m_\mu\rangle \sqrt{1+\langle \tilde{\mathbf{m}}\rangle^2}+\frac{\langle\tilde  m_\mu \rangle}2\sum_{\rho,\sigma}\left(\frac{\delta_{\rho,\sigma}}{\sqrt{1+ \langle\tilde{\mathbf{m}}\rangle^2}}-\frac{\langle \tilde m_\rho\rangle \langle\tilde  m_\sigma\rangle }{({1+ \langle\tilde{\mathbf{m}}\rangle^2})^{3/2}}\right)\frac{\langle\Delta _\rho \Delta_ \sigma\rangle}{N}\\&+
\sum _{\rho} \frac{\langle\tilde  m_\rho \rangle}{\sqrt{1+ \langle\tilde{\mathbf{m}}\rangle^2}}\frac{\langle\Delta_\mu \Delta_\rho\rangle}{N}+\langle \mathcal Q_1^\mu(\tilde{\boldsymbol{m}})\rangle.
\end{split}
\end{equation*}
Here, we used the fact that $\langle \Delta _\mu \rangle =0$ and defined the quantity
\begin{equation}
Q_1^\mu(\tilde{\boldsymbol{m}})=\Big( \langle {\tilde m}_{\mu}\rangle +\frac{\Delta_\mu }{\sqrt N}\Big)\mathcal R_2(\tilde{\boldsymbol{m}})+\frac12\sum_{\rho,\sigma=1}^P\left(\frac{\delta_{\rho,\sigma}}{\sqrt{1+ \langle\tilde{\mathbf{m}}\rangle^2}}-\frac{\langle \tilde m_\rho\rangle \langle\tilde  m_\sigma\rangle }{({1+ \langle\tilde{\mathbf{m}}\rangle^2})^{3/2}}\right)\frac{\Delta_\mu \Delta _\rho \Delta_ \sigma}{N^{3/2}},
\end{equation}
accounting for all the contributions which scale at least as $N^{-3/2}$. In the same fashion, we have
\begin{equation*}
\begin{split}
\langle\tilde{m}_\mu\rangle&\langle \sqrt{1+ \tilde{\mathbf{m}}^2} \rangle=
\langle \tilde m_\mu\rangle \sqrt{1+\langle \tilde{\mathbf{m}}\rangle^2}+\frac{\langle\tilde  m_\mu \rangle}2\sum_{\rho,\sigma}\left(\frac{\delta_{\rho,\sigma}}{\sqrt{1+ \langle\tilde{\mathbf{m}}\rangle^2}}-\frac{\langle \tilde m_\rho\rangle \langle\tilde  m_\sigma\rangle }{({1+ \langle\tilde{\mathbf{m}}\rangle^2})^{3/2}}\right)\frac{\langle\Delta _\rho \Delta_ \sigma\rangle}{N}+\langle Q_2^{\mu}(\tilde{\boldsymbol{m}})\rangle,
\end{split}
\end{equation*}
and
\begin{equation*}
\begin{split}
\langle \tilde m_\mu\tilde m_\rho \rangle=\langle\tilde m_\mu\rangle \langle\tilde m_\rho \rangle+\frac{\langle \Delta_\mu \Delta_\rho\rangle}{N}+\langle \mathcal Q_3^{\mu}(\tilde{\boldsymbol{m}})\rangle.
\end{split}
\end{equation*}
The functions $Q_2 ^\mu $ and $Q_3 ^\mu$ are defined in order to incorporate all the subleading contributions in $N$. Calling $Q^{\mu}(\tilde{\boldsymbol{m}})= Q_1^{\mu}(\tilde{\boldsymbol{m}})-Q_2^{\mu}(\tilde{\boldsymbol{m}})-Q_3^{\mu}(\tilde{\boldsymbol{m}})$, we can finally recast \eqref{eq:derivative} as
\begin{equation}
\label{eq:master1}
\frac{d \langle \tilde{m}_\mu \rangle }{d t} =\beta \sum_\rho \Big(\frac{\langle \tilde m_\rho \rangle}{\sqrt{1+\langle \tilde{\mathbf{m}}\rangle^2}}-\psi^\rho+(1-t)(\psi^\rho)'\Big)\langle \Delta_\mu \Delta_\rho\rangle+N \langle Q^{\mu}(\tilde{\boldsymbol{m}}) \rangle.
\end{equation}
Since $Q^{\mu}(\tilde{\boldsymbol{m}})$ scales as $N^{-3/2}$, the second contribution is subleading (and vanishes in the $N\to\infty$ limit). We now call $\boldsymbol{\tilde{M}} := \lim_{N \to \infty}\langle \mathbf{\tilde{m}}_N \rangle_t$ the thermodynamic value of the global magnetization; notice that at this stage we still allow $\boldsymbol{\tilde{M}}$ to depend on $t$. Clearly, the r.h.s. of Eq. \eqref{eq:master1} is well-defined almost everywhere in the thermodynamic limit (recall that the $\Delta_\mu$ variables have finite covariance). Thus
\begin{equation}
\label{eq:master2}
\lim_{N\to\infty}\frac{d \langle \tilde{m}_\mu \rangle }{d t} =\beta \sum_\rho \Big(\frac{ \tilde M_\rho }{\sqrt{1+ \tilde{\mathbf{M}}^2}}-\psi^\rho+(1-t)(\psi^\rho)'\Big)\langle \Delta_\mu \Delta_\rho\rangle_\infty, \quad {a.e.}
\end{equation}
where the subscript $\infty$ is used to stress that the expectation value is evaluated in the thermodynamic limit. In other words, the series of the $t$-derivatives of the generalized-average magnetization converges to the r.h.s. of the previous equation almost everywhere, thus, invoking Egorov's theorem, the sequence is almost uniformly convergent. As a consequence, we have
$$   \lim_{N \to \infty }\frac{d \langle \tilde{m}^\mu_N \rangle _t}{d t} =\frac{d  }{d t}  \lim_{N \to \infty }\langle\tilde{m}^\mu_N \rangle_t \quad {a.e.} ,$$ thus Eq. \eqref{eq:master2} becomes
\begin{equation}
\frac{d \tilde M_\mu}{dt}=\beta \sum_\rho \Big(\frac{ \tilde M_\rho }{\sqrt{1+ \tilde{\mathbf{M}}^2}}-\psi^\rho+(1-t)(\psi^\rho)'\Big)\langle \Delta_\mu \Delta_\rho\rangle_\infty, \quad {a.e.}
\end{equation}
The function $\boldsymbol \psi$ can be chosen in order to ensure the quantity in round brackets to vanish, which would imply $d\tilde M_\mu /dt=0$ for all $\mu=1,\dots,P$. This implies
\begin{equation*}
(1-t)(\psi^\rho)' -\psi^\rho = -\frac{M_\rho}{\sqrt{1+\tilde{\mathbf{M}}^2}}.
\end{equation*}
Because of our choice, the r.h.s. is independent on $t$. The solution of this differential equation can be easily found recalling that $(1-t)(\psi^\rho)' -\psi^\rho=[(1-t)\psi^\rho]'$, thus
$$
(1-t)\psi^\rho = -\frac{M_\rho}{\sqrt{1+\tilde{\mathbf{M}}^2}}t+ c_\rho,
$$
or in a more transparent form
$$
\psi^\rho =\frac{M_\rho}{\sqrt{1+\tilde{\mathbf{M}}^2}}\frac{c'_\rho-t}{1-t},
$$
where we redefined
$$
c_\rho=\frac{M_\rho}{\sqrt{1+\tilde{\mathbf{M}}^2}}c'_\rho.
$$
Now, since $\boldsymbol \psi$ in the interpolating model always appears through the product $(1-t)\psi^\rho$ and since we have to recover the original model \eqref{original_model}, one has that $c'_\rho = 1$ for all $\rho=1,\dots,P$, therefore leaving us only with
\begin{equation} \label{eq:psimu}
\psi^\rho= \frac{\tilde M_\rho}{\sqrt{1+ \tilde{\mathbf{M}}^2}}.
\end{equation}
This proves our assertion.
\end{proof}
We can now state the following \medskip
\begin{Theorem}
In the thermodynamic limit and under the self-average properties of the order parameters, the intensive pressure of the ``relativistic'' Hopfield model with correlated patterns
 (\ref{Hopfield-relativistico}) can be written in terms of the $P$ Mattis magnetizations as
\begin{equation}
\label{FreeFinal}
F_{\beta, a}=\log 2 +  \mathbb E \log \cosh \left(\beta 
\frac{ \sum_\mu \xi^\mu (\boldsymbol{ X} \boldsymbol{M})_\mu}{\sqrt{1+  \boldsymbol{M}^T\boldsymbol{ X}\boldsymbol{M} }} \right) + \frac{\beta}{\sqrt{1+  \boldsymbol{M}^T\boldsymbol{ X}\boldsymbol{M}}}.
\end{equation}
The associated self-consistency equations read
\begin{equation}
\label{finale}
M_{\mu}  =\frac{(1+\boldsymbol{M}^T\boldsymbol{X} \boldsymbol{M})\mathbb E \,{\xi}^{\mu} \tanh\Big( \beta \frac{\sum_\rho {\xi}^\rho  (\boldsymbol{X} \boldsymbol{ M})_\rho}{\sqrt{1+\boldsymbol{M}^T\boldsymbol{X} \boldsymbol{M}}} \Big)}{1+\sum_\nu (\boldsymbol M^T \boldsymbol X)_{\nu }\mathbb E \,{\xi}^{\nu} \tanh\Big( \beta \frac{\sum_\rho {\xi}^\rho  (\boldsymbol{X} \boldsymbol{ M})_\rho}{\sqrt{1+\boldsymbol{M}^T\boldsymbol{X} \boldsymbol{M}}} \Big)} , \quad \ \forall \mu=1,...,P.
\end{equation}
\end{Theorem}
\begin{proof}
As anticipated, the strategy is to get the pressure of the model (\ref{Hopfield-relativistico}), by exploiting (\ref{nuova}). As a first step, we evaluate the interpolating pressure (\ref{alpha-intes}) at $t=0$; this can be done straightforwardly as it corresponds to a one-body system:
%
\begin{eqnarray}
\nonumber
\bar{F}_N(0) &=& \frac{1}{N} \log \Big[   \sum_{ \{ \boldsymbol \sigma \}} \exp\Big(\beta \sum_{\mu=1}^{P} \psi^{\mu} \sum_{i=1}^{N} \tilde{\xi}_i^{\mu}\sigma_i \Big) \Big]\\
 & =& \log{2} + \frac{1}{N} \sum_{i=1}^{N} \log{ \cosh \Big (\beta \sum_{\mu=1}^{P} 
 \psi^{\mu} \tilde{\xi_i}^{\mu} \Big) }\label{Cauchy}\\
&=&
\nonumber
\log{2} + \Big \langle \log\cosh\Big( \beta \sum_{\mu=1}^{P}\psi^{\mu}{\tilde{\xi}^{\mu}}_{i} \Big) \Big \rangle_{\tilde{\boldsymbol{\xi}}},
\end{eqnarray}
where $\langle  \cdot  \rangle_{\boldsymbol {\tilde \xi } }$ represents the empirical average over the $\boldsymbol{ \tilde \xi}$ and, as $N \to \infty$, it can be replaced by the expectation $\mathbb{E}_{\tilde {\boldsymbol \xi}}(\cdot)$. Next, let us move to the calculation of the $t$-derivative, that is,
\begin{equation} 
\frac{d\bar{F}_N(t)}{dt} 
= 
\beta \Big \langle  \sqrt{1+ \tilde{\mathbf{m}}_N^2}- \sum_{\mu=1}^{P}\psi^{\mu} \tilde{m}_N^\mu \Big \rangle_{t}.
\label{punto}
\end{equation}
In the thermodynamic limit, the generalized measure concentrates and, in particular, for almost every value of $\beta \in \mathbb{R}^+$ \cite{Guerra1}
\begin{equation}
    \label{punto1}
\lim_{N \to \infty}  \frac{d\bar{F}_N(t)}{dt} = \frac{d\bar{F}(t)}{dt} =
\beta \Big(  \sqrt{1+  \tilde{\boldsymbol {M}} ^2}- \sum_{\mu=1}^{P}\psi^{\mu}  \tilde{M}_\mu  \Big ).
\end{equation}
Here, we stress that, due to the Lemma \ref{prop:constant}, we directly used the equilibrium value of the order parameters, which are $t$-independent, and the integral in \eqref{nuova} therefore turns out to be trivial.
%
%
%
%
Recalling the expression for $\psi^{\mu}$ in \eqref{eq:psimu}, we get
\begin{equation} \label{delta}
\frac{d \bar{F}(t)}{dt} = \frac{\beta}{\sqrt{1+\boldsymbol{\tilde{M}}^2}}.
\end{equation}
By plugging \eqref{Cauchy} and \eqref{delta}  into \eqref{nuova}, we get the intensive pressure in terms of the rotated order parameters: 
\begin{equation}
\label{free-final}
F(\boldsymbol{\tilde M})=\log 2 +\mathbb{E}_{\tilde {\boldsymbol \xi}}  \log \cosh \Bigg(\beta \tilde{\boldsymbol{\xi}}^T\cdot
\frac{ \boldsymbol{\tilde{M}}}{\sqrt{1+  \boldsymbol{\tilde{M}}^2 }} \Bigg) + \frac{\beta}{\sqrt{1+  \boldsymbol{\tilde{M}}^2}}.
\end{equation}
The expression \eqref{FreeFinal} can be obtained by rotating back the order parameters in the original space. Since the extremization of the intensive pressure w.r.t. the real Mattis magnetizations $\boldsymbol M$ is equivalent to the extremization w.r.t. to the rotated ones, we can directly impose the extremality condition on the \eqref{free-final}, thus obtaining the condition
\begin{equation}
\tilde M_\mu=\mathbb E_{\tilde{\boldsymbol \xi}} \left \{ \tanh \left (\beta \frac{ \tilde{\boldsymbol{\xi}}^T\cdot \tilde{\boldsymbol{M}}}{ \sqrt{1 + \tilde{\boldsymbol{M}^2}}} \right) \left [\tilde{\xi}^{\mu} \left(1 + \tilde{\boldsymbol{M}^2} \right) - \tilde{M}_{\mu} \left( \tilde{\boldsymbol{\xi}}^T \cdot \tilde{\boldsymbol{M}} \right) \right]\right \}.
 \end{equation}
We stress that the value of $\tilde M_\mu$ does not depend on the specific realization of the digital patterns $\xi^\mu$, since the r.h.s. of the previous equation is averaged over the quenched noise. Due to this fact, we can rearrange the r.h.s. as
\begin{equation}
\begin{split}
\tilde M_\mu = &\mathbb E_{\tilde{\boldsymbol \xi}} \left(\tanh \Bigg(\beta \frac{ \tilde{\boldsymbol{\xi}}^T\cdot \tilde{\boldsymbol{M}}}{ \sqrt{1 + \tilde{\boldsymbol{M}^2}}} \Bigg) [\tilde{\xi}^{\mu} (1 + \tilde{\boldsymbol{M}^2}) - \tilde{M}_{\mu}( \tilde{\boldsymbol{\xi}}^T \cdot \tilde{\boldsymbol{M}}) ]\right)=\\=&\,(1+\tilde{\boldsymbol M}^2)\mathbb E_{\tilde{\boldsymbol \xi}}\, \tilde \xi^\mu\tanh \Bigg(\beta \frac{ \tilde{\boldsymbol{\xi}}^T\cdot \tilde{\boldsymbol{M}}}{ \sqrt{1 + \tilde{\boldsymbol{M}^2}}} \Bigg)-\tilde M_\mu \,\sum_\nu \tilde M_\nu\mathbb E_{\tilde{\boldsymbol \xi}}\, \tilde \xi^\nu\tanh \Bigg(\beta \frac{ \tilde{\boldsymbol{\xi}}^T\cdot \tilde{\boldsymbol{M}}}{ \sqrt{1 + \tilde{\boldsymbol{M}^2}}} \Bigg).
\end{split}
\end{equation}
Further, moving the second term in the r.h.s. to the l.h.s. and collecting $\tilde M_\mu$, we obtain
\begin{equation}
\tilde M_\mu=\frac{(1+\tilde{\boldsymbol M}^2)\mathbb E_{\tilde{\boldsymbol \xi}}\, \tilde \xi^\mu\tanh \Big(\beta \frac{ \sum_\rho \tilde{{\xi}}_\rho \tilde{{M}}_\rho}{ \sqrt{1 + \tilde{\boldsymbol{M}^2}}} \Big)}{1+\sum_\nu \tilde M_\nu\mathbb E_{\tilde{\boldsymbol \xi}}\, \tilde \xi^\nu\tanh \Big(\beta \frac{ \sum_\rho \tilde{{\xi}}_\rho \tilde{{M}}_\rho}{ \sqrt{1 + \tilde{\boldsymbol{M}^2}}} \Big)}.
\end{equation}
Then, (\ref{finale}) follows by rotating back the order parameters in the original pattern space.
\end{proof}
We stress that, since the l.h.s. of \eqref{FreeFinal} does not depend on the specific realization of the digital patterns $\boldsymbol \xi$, we get a further proof of Proposition \ref{prop1}.

\section{Numerical solution and phase diagrams}\label{sec:numerics}

In this Section, we report the numerical solutions for the self-consistency equations \eqref{finale}. To this aim, we fix $P$ and $a\in[0,1]$ and solve the self-consistency equations as a function of $T=\beta^{-1}$ with a fixed-point iteration method. 
The solutions obtained by setting as initial configuration a retrieval state (i.e. $\boldsymbol{M}^T=(1,0,..,0)$, without loss of generality), and by tuning the parameters $T$ and $a$ can be seen in the left panel of Figure \ref{SCsol_Fmin}. 
\begin{figure}[!ht]
    \centering
    \includegraphics[width=1.0\textwidth]{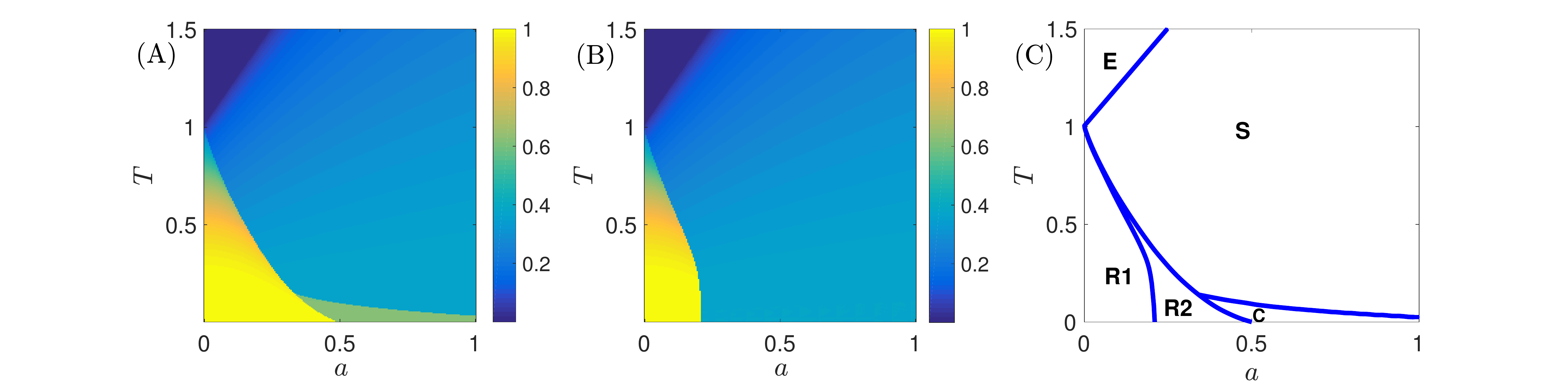}
    \caption{Panel (A): numerical solution of the self-consistency equation for $P=5$ and using the state $\boldsymbol{M}^T=(1,0,..,0)$ as a starting point. Panel (B): numerical solution of the self-consistency equation for $P=5$ obtained by considering different starting states and selecting, among the related solutions, the one corresponding to the largest pressure. The heat map is realized by considering the value of the largest Mattis magnetization. Panel (C): complete phase diagram for the correlated model obtained from the two previous panels. (E) represents the ergodic phase where $M_\mu=0$ $\forall \mu=1;..,P$, (S) is the spin-glass phase where $M_\mu=m\ne 0$ $\forall \mu=1,..,P$ and (R1) and (R2) are the recall phases of the model. In particular (R1) corresponds to the region where the pure states maximize the pressure while in (R2) the pure states are local pressure maxima.}
    \label{SCsol_Fmin}
\end{figure}
By inspecting this plot, we can notice the existence of different regions, analogously to the case of the classical Hopfield model with correlated patterns (Fig. \ref{phasediagram}):
 \begin{itemize}
  \item For high level of noise the system is ergotic (E) and the only stable solution is the one where any entry of the magnetization vanishes, that is $\boldsymbol{M}=\mathbf{0}$;
  \item At smaller temperature values the system evolves to the symmetrical phase (S) in which the solution is of the form $\boldsymbol{M}^T=(m,m,m,...,m)$ with $ \ m \neq 0$; 
  \item Decreasing the temperature and taking $a$ small enough, the system enters into the retrieval phase (R) where the stable solution has only one non-zero component corresponding to the retrieved pattern; 
  \item For low temperature and large values of the parameter $a$, we can see a hierarchical phase (C) where the solution displays several non-vanishing magnetizations. More precisely, there is one overlap (which is used as initial stimulus) with maximum value, while the remaining magnetizations decay symmetrically in the distance w.r.t. to the initial stimulus, until they vanish. In particular, we note that overlaps that are at the same distance w.r.t. the starting signal have the same value.
\end{itemize}

Now, in order to get an overall picture of the system behavior we proceed with the construction of its phase diagram. To this aim, we solve numerically Eq.~\eqref{finale} for different initial configurations and then, for the related solutions we compute the pressure. Then, selecting as solutions those for which the pressure is maximal, we obtain Fig.~\ref{SCsol_Fmin} (central panel). \\ 
\\By comparing the panels (A) and (B) of Fig.~\ref{SCsol_Fmin} we get panel (C). In particular, we can see that the retrieval region observed for values of $T$ and $a$ relatively small can be further split in a pure retrieval region (R1) where pure states are global maxima for the intensive energy, and in a mixed retrieval region (R2) where pure states are local minima, yet their attraction basin is large enough for the system to end there if properly stimulated. In these two regions, the network behaves like a Hopfield network and the patterns can be recovered. On the other hand, even if the temperature $T$ remains low, for values $a>0.5$ the pure state regime is no longer achievable.

Now, if we look at the figure more carefully, we can observe that the ergodic phase occurs beyond a certain temperature $T_c(a)$, whose value increases with the correlation parameter $a$. Just below this critical temperature the system enters the symmetrical phase where the effect of the temporal correlation is strong enough for each pattern to align the same fraction of neurons. 
We can analytically determine the transition line (see Appendix \ref{sec:Tc_line} for the analytical derivation) as
\begin{equation}\label{Tc1}
T_c(a)=1+2a
\end{equation}
by Taylor expanding the right hand side of eq.~\eqref{finale}, recalling that in the ergodic phase $\boldsymbol M = \boldsymbol 0$. Eq.~\eqref{Tc1} indicates a continuous transition to the fully symmetric phase, where there is no significant alignment of the spins in the direction of one particular pattern, but still a certain degree of local freezing. We can therefore state that  if $T>T_c$ the only solution is $\boldsymbol{M}=\boldsymbol{0}$ while, if $T<T_c$ there exist solutions $\boldsymbol{M}\ne\boldsymbol{0}$.
The line $T_c$ analytically found is consistent with the numerical solution of the Eq. \eqref{finale}.

\begin{figure}[tb]
    \centering
    \includegraphics[width=1.0\textwidth]{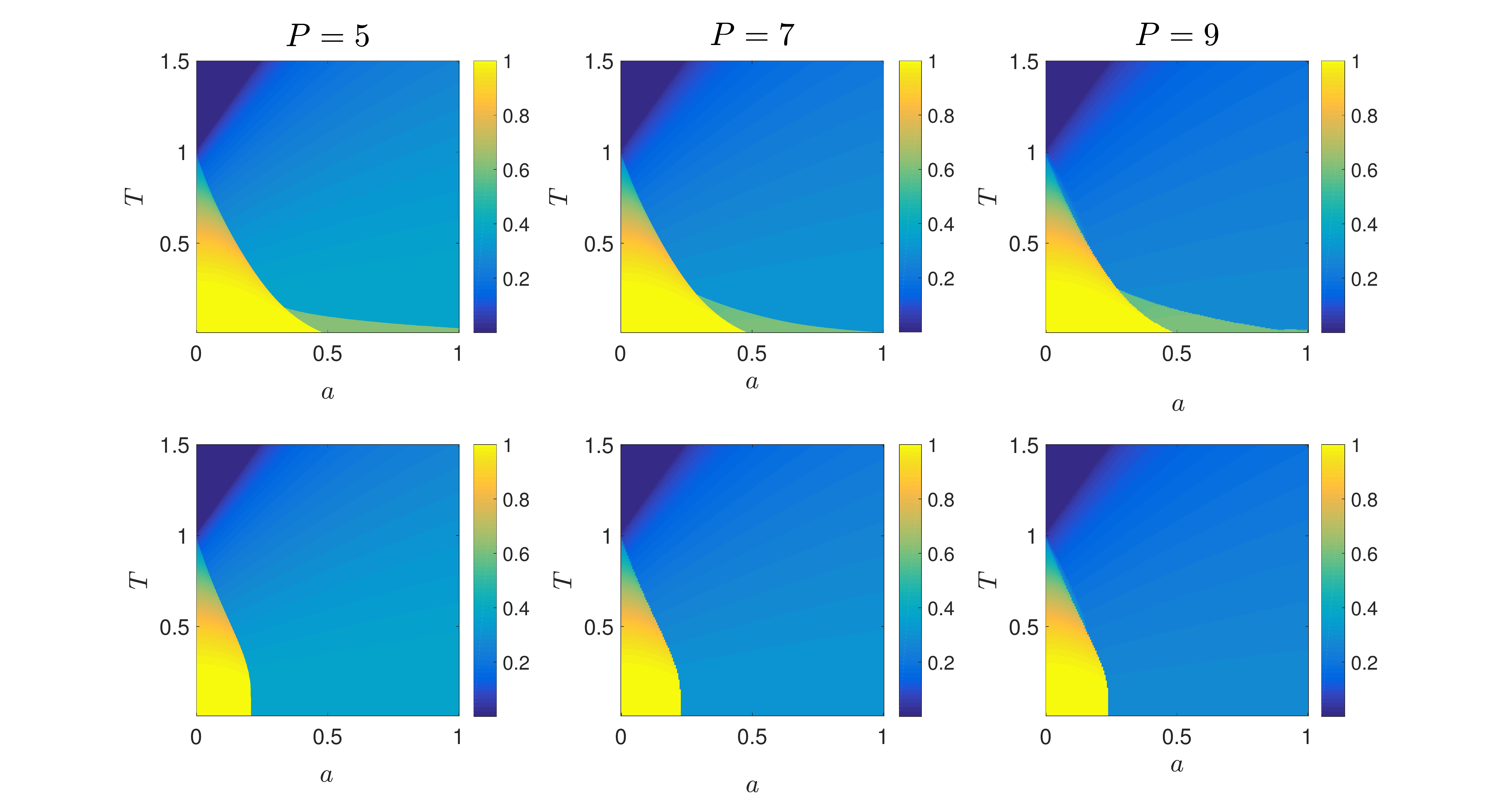}
    \caption{First line: solutions of the self-consistency equation obtained by taking as initial condition the pure state, as function of $T$ and $a$. Second line: selected solutions corresponding to the maximum of the intensive pressure. In both cases different values of $P$ ($P=5,7,9$) are considered. The heat map is realized by considering the value of the largest Mattis magnetization.}
    \label{P579}
\end{figure}

We now focus on the parameter $P$. In Fig.~\ref{P579} we show (in the first line) the solutions of the self-consistency equations for different values of $P$ ($P=5,7,9$) when the initial state is the pure configuration $\boldsymbol{M}^T=(1,0,...,0)$ and when different initial states are tested (second line) with the solution giving the maximal pressure. From this picture, we can notice that, as $P$ increases, the correlated (C) and the retrieval (R1) phases get wider.
Let also consider the following initial states:
\newline  
\\$\boldsymbol{M}^T=\frac{1}{2}(1,1,1),\quad \mbox{if } P=3$
\\
\\$\boldsymbol{M}^T=\frac{1}{8}(5,3,1,1,3),\quad \mbox{if } P=5$
\\
\\$\boldsymbol{M}^T=\frac{1}{32}(19,13,3,1,1,3,13),\quad \mbox{if } P=7$
\\
\\$\boldsymbol{M}^T=\frac{1}{128}(77,51,13,3,1,1,3,13,51),\quad \mbox{if } P=9.$

These vectors correspond to the states in the correlated region where the Mattis magnetizations reach a hierarchical structure in which the largest element of the vector, for example $M_1$, corresponds to the initial stimulus and the other elements are symmetrically decreasing i.e. $M_1\ge M_2=M_P\ge M_3=M_{P-1}\ge...\ge M_{\frac{P+1}{2}}=M_{\left(\frac{P+1}{2}\right)+1}$.
If we consider these configurations as initial conditions for solving Eqs. \eqref{finale} for varying $a$ and $T$, we get the Figure \ref{P579_M0corr}.
It can be seen that, for each value of $P$, part of the retrieval region is absorbed by the correlated one which becomes more and more predominant as $P$ increases. This means that starting from such a configuration, retrieval gets harder and harder. 
 \begin{figure}[tb]
    \centering
    \includegraphics[width=1.0\textwidth]{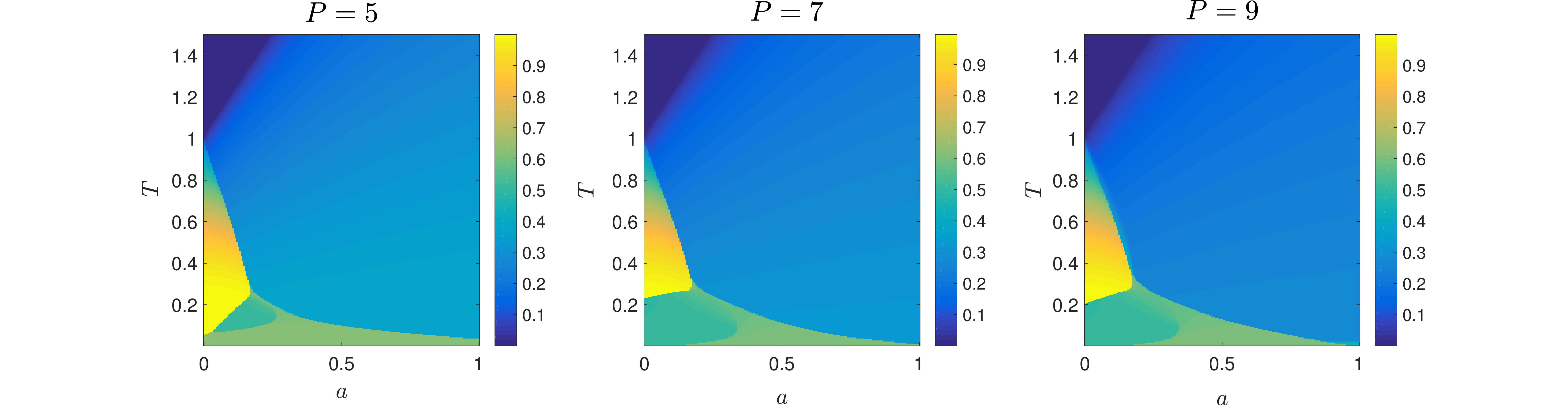}
    \caption{Solutions of the self-consistency equation for $P=5,7,9$ starting from a correlated state. The heat map is realized by considering the value of the largest Mattis magnetization.}
    \label{P579_M0corr}
\end{figure}
 \begin{figure}[tb]
	\centering
	\includegraphics[width=0.8\textwidth]{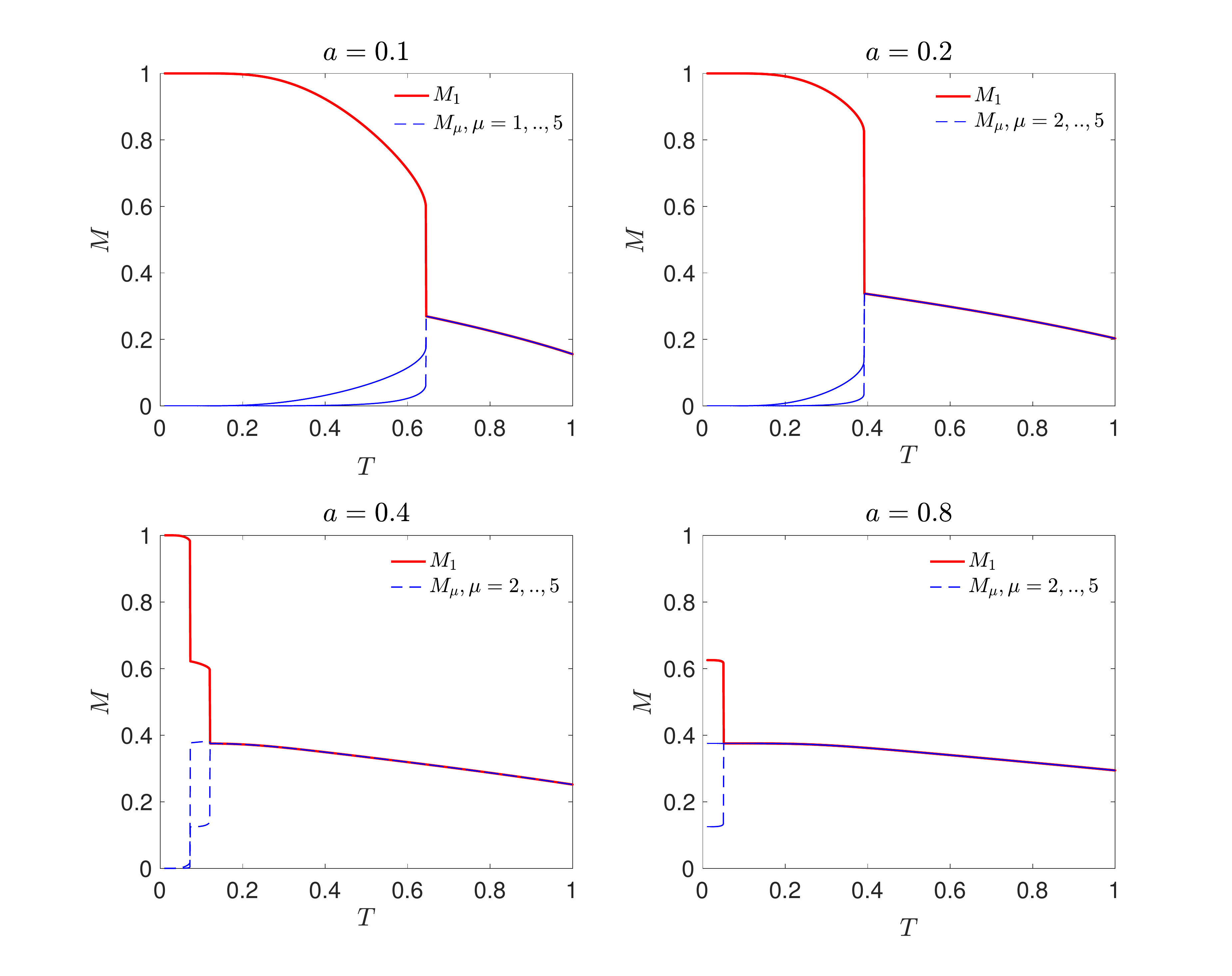}
	\caption{These panels show the numerical solution for the self-consistency equation ~\eqref{finale} in case $P = 5$ for different values of $a$; self-consistency equation are solved by taking as initialization the pure state $\mathbf{M}^T=(1,0,...,0)$. The red line is the one related to the magnetization of the first pattern, while the four blue dashed lines indicate the magnetizations of the other patterns. An analogous behavior has been observed also for $P =7$ and $P =9$.} 
	\label{a_fissato}
\end{figure}
Let us now fix the correlation parameter $a$ and see the evolution of the magnetization entries as the temperature is varied. 
As shown in Fig.~\ref{a_fissato}, under the assumption that the first pattern acts as a stimulus we get that $M_1$ remains the largest component until the symmetrical phase is reached where $M_{\mu}=m$ for any $\mu$. This behavior is robust with respect to $P$. 

 As a final investigation, we consider the robustness of the retrieval solution by solving the self-consistent equation (eq.~\eqref{finale}) starting from an initial configuration given by $\boldsymbol M ^T= (1 - \delta, \delta, ..., \delta)$, namely, a pure state affected by a degree of noise tunable by the parameter $\delta$, and homogenously spread over all the magnetization entries. The resulting solution is shown in Fig.~\ref{rumoreP5}: as $\delta$ is increased the correlated and the retrieval regions progressively breaks down into a symmetric-like region and the most sensitive area is the one corresponding to large value of $a$. This suggests that the correlated region is relatively unstable.  

 \begin{figure}[bt]
    \centering
    \includegraphics[width=1.0\textwidth]{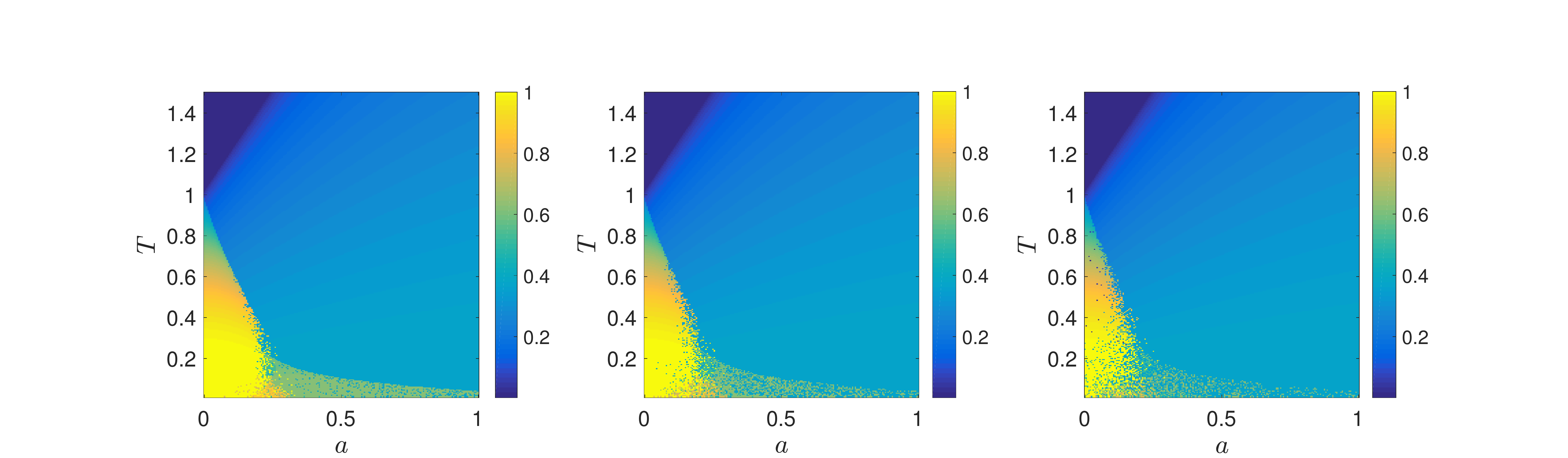}
    \caption{Solution of the self-consistency equation for $P = 5$ evaluated by using as initial state the vector $\boldsymbol{M }^T= (1 - \delta, \delta, ..., \delta)$, where $\delta = 0.15$ (leftmost panel),  $\delta = 0.20$ (central panel), and $\delta = 0.25$ (rightmost panel). Notice that the region mostly affected by the noise $\delta$ is the one corresponding to the correlated phase. The heat map is realized by considering the value of the largest Mattis magnetization.}
    \label{rumoreP5}
\end{figure}

\section{Conclusions} \label{sec:conclusions}

In this work we considered the ``relativistic'' Hopfield model in the low-load regime and we allowed for possible temporal correlations among memories. This is accomplished by revising the Hebbian couplings as $J_{ij} = \frac{1}{N} \sum_{\mu} [\xi_i^{\mu} \xi_j^{\mu} + a(\xi_i^{\mu+1} \xi_j^{\mu} + \xi_i^{\mu-1} \xi_j^{\mu})]$ in agreement with previous experimental results \cite{Amit2,Miyashita-1988, MiyashitaChang-1988}. The model parameters are the extent of temporal correlation $a$ and the external noise $T$.
We first addressed the model analytically -- showing the existence of the thermodynamic limit and calculating the free-energy, whence the self-consistency equations for the order parameters -- and numerically -- deriving a phase diagram displaying different qualitative behaviors of the system in the $(T,a)$ space. In particular, beyond the ergodic and the retrieval regions, we highlight the emergence of a so-called symmetric region, characterized by states where the overlap $m_{\mu}$ between the $\mu$-th memory and the neural configuration is $M_{\mu}=m>0$ for any $\mu$, and a correlated region, characterized by states with a hierarchical arrangement of overlaps which mirrors the temporal correlation among memories, that is $M_1 \ge M_2=M_P \ge M_3 =M_{P-1} ... $ when the stimulate memory is $\boldsymbol{\xi^1}$. The appearance of both the symmetric and the correlated region is a genuine effect of correlation and, in fact, they were also found in the classic Hopfield model with temporal correlation \cite{GriniastyTsodyksAmit-1993,Cugliandolo-1993,CugliandoloTsodyks-1994,Agliari-Dantoni}.
The ``relativistic'' nature of the model under investigation just makes these regions wider in such a way that, even at relatively small values of temperature the system may relax in a symmetric state (when $a$ is relatively small) or in a correlated state (when $a$ is relatively large). 
Therefore, in order to highlight the effect of correlation, the relativistic model is more suitable than the classical model.

\appendix

\section{Self-average of the pressure} \label{sec:automedia}
In this Appendix we report the proof of Proposition 1, which is reported hereafter for completeness 
\begin{customlemma}{1}
In the thermodynamic limit, the self-average property of the intensive pressure holds, i.e.
\begin{equation}
\lim_{N \to \infty} \mathbb{E} \left \{ F_{N,\beta,a}(\boldsymbol{\xi})-\mathbb{E} [F_{N,\beta,a}( \boldsymbol{\xi})]  \right \}^2 =0
\end{equation}
where $$F_{N,\beta,a}(\boldsymbol{\xi})=\frac{1}{N}\log Z_{N,\beta,a}(\boldsymbol{\xi})=\frac{1}{N}\log\sum_{\{ \boldsymbol \sigma \}}e^{- \beta H_{N,a}(\boldsymbol \sigma| \boldsymbol{\xi})}$$ and 
\begin{equation}
     \label{Hopfield-relativistico}
 H_{N,a}(\boldsymbol \sigma| \boldsymbol{\xi}) = - N \sqrt{1+ \frac{1}{N^2}\sum_{\mu=1}^{P}\sum_{i,j=1}^{N,N}{\sigma_i \sigma_j\tilde{\xi_i^{\mu}}\tilde{
\xi_j^{\mu} }}}.
\end{equation}
\end{customlemma}
To this aim we shall adapt the martingale method originally developed for Hopfield networks \cite{pastur1,pastur2,tirozzi}.
Notice that, for simplicity of notation, we will indicate $F_{N,\beta,a}(\boldsymbol{\xi})$ with $F_N$ in the proposition proof. 
\begin{proof}
Let $\mathcal{R}_k^N$ be the $\sigma$-algebras generated by the sets $\{\tilde{\xi}_i^\mu\}_{i\ge k}^{\mu=1,...,P}$ where  $k=1,...,N$ 
and $F_N^{k}$ be the conditional expectation of $F_N$ with respect to $\mathcal{R}_k^N$ i.e 
\begin{equation} 
F_N^k=\mathbb{E}[F_N|\mathcal{R}_k^N]=\mathbb{E}_{<k}[F_N]
\end{equation}
where 
\begin{equation}\label{conditional_average}
\mathbb{E}_{<k}[\ \cdot\ ]:=\frac{1}{2^{P(k-1)}}\sum_{\{\xi^\mu_i\}_{i<k}^{\mu=1,...,P}}(\ \cdot \ ).
\end{equation}
Being $\{\tilde{\xi}_i^\mu\}_{i\ge k}^{\mu=1,...,P}$ a descending family of $\sigma$-algebras then also the $\sigma$-generated by this family will be such that $\mathcal{R}_{k+1}^N\subseteq\mathcal{R}_k^N$ $\forall k$. As a consequence, the sequence constituted by the stochastic variables $F_N^k$ and by the family $\left\{\mathcal{R}_k^N\right\}_{k=1}^N$ fulfill the following  Martingale property:
\begin{equation}
\mathbb{E}[F^k_N|\mathcal{R}_l^N]=\begin{cases} F^k_N, & \mbox{if }k>l \\ F^l_N, & \mbox{if }k\le l. \end{cases}
\end{equation}
It is possible to notice that
\begin{equation*}
\sum_{k=1}^N F_N^k-F_N^{k+1}=F_N^1-F_N^{N+1}= \log Z_N - \mathbb{E} (\log Z_N), 
\end{equation*}
therefore if we define $\Psi_k:=F_N^k-F_N^{k+1}$ we get
\begin{equation}
F_N-\mathbb{E}[F_N]=\frac{1}{N}\sum_{k=1}^N\Psi_k.
\end{equation}
As a consequence we can write 
\begin{equation}
\mathbb{E}\left[(F_N-\mathbb{E}F_N)^2\right]=\frac{1}{N^2}\sum_{k=1}^N\mathbb{E}[\Psi^2_k]+\frac{2}{N^2}\sum_{k=1}^N\sum_{l=k+1}^N \mathbb{E}[\Psi_k \Psi_l] 
\end{equation}
where in the right hand side we split the diagonal and the off-diagonal contributions. 
\\Using the properties of the conditional expectation, we have
\begin{equation*}
\begin{split}
\mathbb{E}[\Psi_k \Psi_l]&=\mathbb{E}\left[\mathbb{E}\left[\Psi_k \Psi_l|\mathcal{R}_l^N\right]\right]=\mathbb{E}\left[\Psi_l\mathbb{E}[\Psi_k |\mathcal{R}_l^N]\right]=\mathbb{E}\left[\Psi_l\mathbb{E}[F_N^k-F^{k+1}_N |\mathcal{R}_l^N]\right]\\&=\mathbb{E}\left[\Psi_l(F_N^l-F^{l}_N )\right]=0
\end{split}
\end{equation*}
thus, in order to obtain the result, it is enough to show that
\begin{equation}\label{lim}
\lim_{N\to\infty}\frac{1}{N^2}\sum_{k=1}^N\mathbb{E}[\Psi^2_k]=0.
\end{equation}
We will prove eq. \eqref{lim} by proving that $\mathbb{E}[\Psi^2_k]$ is bounded for each $k$.

To this aim, let us define the following interpolating functions
\begin{equation}
\begin{split}
&\Phi_k(t):= - N \sqrt{1+ \frac{1}{N^2}\sum_{\mu=1}^{P}\sum_{i\ne j\ne k}^{N,N}{\sigma_i \sigma_j\tilde{\xi_i^{\mu}}\tilde{\xi_j^{\mu} }}+\frac{t}{N}\sum_{\mu=1}^{P}\sum_{i\ne k}^{N}{\sigma_i \sigma_k\tilde{\xi_i^{\mu}}\tilde{\xi_k^{\mu} }}}.
\\&g_k(t):=\log Z_N(\Phi_k(t))-\log Z_N(\Phi_k(0))
\end{split}
\end{equation}
the interpolating function $\Phi_k$ returns the original Hamiltonian when $t = 1$. 
\\By applying the conditional average \eqref{conditional_average} to $g_k(1)$ and comparing with the definition $\Psi_k:=F_N^k-F_N^{k+1}$ we have the following identity
\begin{equation}
\Psi_k =\mathbb{E}_{<k}(g_k(1))-\mathbb{E}_{<k+1}(g_k(1))
\end{equation}
from which we get 
\begin{equation}\label{disug}
\mathbb{E}\Psi^2_k\le 2\mathbb{E}[g_k(1)^2]. 
\end{equation}

By making a Taylor expansion of $g_k(t)$ around $t = 0$ and $t=1$ we get,
\begin{equation}\label{Taylor1}
g_k(t)=g_k(0)+g_k'(0)t+g''_k(\eta)
\end{equation}
\begin{equation}\label{Taylor2}
g_k(t)=g_k(1)+g_k'(1)(t-1)+g''_k(\eta)
\end{equation}
where $0\le\eta\le1$. 
By calculating \eqref{Taylor1} in $t=1$ and \eqref{Taylor2} in $t=0$ and observing that $g_k(0)=0$ we obtain
\begin{equation}\label{Taylor1.1}
g_k(1)=g_k'(0)+g''_k(\eta)
\end{equation}
\begin{equation}\label{Taylor2.1}
g_k(1)=g_k'(1)-g''_k(\eta).
\end{equation}
From the equations \eqref{Taylor1.1} and \eqref{Taylor2.1} we therefore have that 
\begin{equation}
2|g_k(1)|\le|g_k'(0)|+|g_k'(1)|.
\end{equation}
Now, if we define $R_k:=\sum_{\mu=1}^{P}\sum_{i\ne k}^{N}{\sigma_i \sigma_k\tilde{\xi}_i^{\mu}\tilde{\xi}_k^{\mu}}$ we can observe that $g_k'(1)$ is just the thermal average of this term. Moreover, recalling that the variables $\tilde{\xi}$ are uncorrelated and that we are in the low load regime, one can state that $|g'_k(1)|$ is bounded. With similar arguments it is possible to conclude that $|g'_k(0)|$ is also bounded by a constant. We can then write that 
\begin{equation}
|g_k(1)|\le|g_k'(0)|+|g_k'(1)|\le C. 
\end{equation}
Now, recalling \eqref{lim} and applying the bound just found we get 
\begin{equation}
\lim_{N\to \infty}\mathbb{E}\left[(F_N-\mathbb{E}F_N)^2\right]=\lim_{N\to\infty}\frac{1}{N^2}\sum_{k=1}^N\mathbb{E}[\Psi^2_k]\le\lim_{N\to \infty} \frac{CN}{N^2}=0.
\end{equation}
\end{proof}

\section{Properties of the correlation matrix and pattern rotation} \label{app:Conti1}
In this Appendix, we will give more details about the temporal correlation matrix with periodic boundary conditions. For the sake of clearness, we report here the matrix $\boldsymbol X$ for the model under consideration:
\begin{eqnarray} 
\boldsymbol{X}  = \left(
\begin{array}{ccccc}
1 & a  & \cdots &0& a  \\
a & 1& & 0 & 0  \\
\vdots & \vdots & \ddots & \vdots & \vdots  \\
0 & 0 & \dots & 0 & a  \\
a & 0 & \cdots & a & 1  \\
\end{array}
\right).
\end{eqnarray}
Trivially, the matrix is symmetric. We now compute its spectrum. To do this, it is useful to introduce a related matric $\boldsymbol X'$ which is obtained by simply removing the boundary conditions, i.e.
\begin{eqnarray} 
\boldsymbol{X}'  = \left(
\begin{array}{ccccc}
1 & a  & \cdots &0& 0  \\
a & 1& & 0 & 0  \\
\vdots & \vdots & \ddots & \vdots & \vdots  \\
0 & 0 & \dots & 0 & a  \\
0& 0 & \cdots & a & 1  \\
\end{array}
\right).
\end{eqnarray}
Now, let as denote (for fixed $P$) the characteristic polynomials of these two matrix, i.e. $D_P(\lambda):=\mbox{ det}_P(\boldsymbol{X}- \lambda \mathbb{I })$ and $D'_P(\lambda):=\mbox{ det}_P(\boldsymbol{X}'- \lambda \mathbb{I })$, where the subscript $P$ stands for the fact that we are taking the determinant of $P\times P$-dimensional matrices. By direct computation, we can see that the two characteristic polynomials are related with the following identity:
\begin{equation}\label{eq:DDprime}
D_P (\lambda)= D'_P(\lambda) -a^2 D'_{P-2} (\lambda)+2(-1)^{P+1} a^P.
\end{equation}

The determination of the characteristic polynomial in absence of the boundary conditions is much more tractable because of the structure of the matrix $\boldsymbol X'$. Furthermore, it is easy to check that it satisfies recurrence relation of Fibonacci-type:
\begin{equation}
D'_P (\lambda)-(1-\lambda) D'_{P-1}(\lambda) +a^2 D_{P-2}'(\lambda)=0,
\end{equation}
and can be solved in a standard way. Indeed, we can look for solutions of the form
\begin{equation}
D'_P(\lambda)=A r_+^P +B r_-^P,
\end{equation}
where $r_\pm$ are the roots of the equation $r^2-(1-\lambda)r+a^2=0$.
Now, setting $\phi=(1-\lambda)/2$ and using the initial conditions $D_1'(\lambda)=2 \phi$ and $D_2'(\lambda)=4\phi^2-a^2$, we obtain $A=\frac{1}{2} \left(1+\frac{\phi}{\sqrt{\phi^2-a^2}}\right)$ and $B=\frac{1}{2} \left(1-\frac{\phi}{\sqrt{\phi^2-a^2}}\right)$. Moreover, noticing that $r_\pm=\phi\pm\sqrt{\phi^2-a^2}$, we get the expression 
\begin{equation*}
D'_P(\lambda)= \frac{1}{2} \Big(1+\frac{\phi}{\sqrt{\phi^2-a^2}}\Big)(\phi+\sqrt{\phi^2-a^2})^P+\frac12 \Big(1-\frac{\phi}{\sqrt{\phi^2-a^2}}\Big)(\phi-\sqrt{\phi^2-a^2})^P.
\end{equation*}
We checked this solution up to $P=30$.
Due to the relation \eqref{eq:DDprime}, we can use this result to get
\begin{equation}
D_{P}(\lambda)= ( \phi - \sqrt{\phi^{2}-a^2})^P+( \phi + \sqrt{\phi^2 - a^2})^P -2(-1)^Pa^P.
\end{equation}
Now, setting $D_P(\lambda)=0$ and replacing $\phi=-a \cos \theta$, we finally get $\theta = 2\pi k/P$, which means that the eigenvalues are fixed to
\begin{equation}
\lambda_k= 1+2a \cos \frac{2\pi k}{P},
\end{equation}
with $k=0,1,\dots,P-1$. We stress that, for $a<1/2$, the entire spectrum is positive.\par
Let us now consider the coupling matrix
\begin{equation}
J_{ij}=\frac1{2N}\sum_{\mu,\nu=1}^P \xi^\mu_i X_{\mu\nu}\xi^\nu_j.
\end{equation}
The matrix $\boldsymbol{X}$ is real and symmetric, thus it is diagonalizable by means of an orthogonal transformation. Then, we can write $\boldsymbol{D}=\boldsymbol{U}^{T}\boldsymbol{X}\boldsymbol{U}$, where $\boldsymbol{D}$ is the diagonal matrix composed by placing the eigenvalues of $\boldsymbol X$ on the main diagonal. Since for $a<1/2$ the spectrum is positive, the matrix $\boldsymbol{D}$ admits a square root $\sqrt {\boldsymbol{D}}$. As a consequence, the interaction matrix $\boldsymbol{J}$ can be rewritten as
\begin{equation}
J_{ij}=\frac1{2N} \sum _{\rho=1}^P \tilde\xi^\rho _i \tilde\xi^\rho _j,
\end{equation}
with the rotated patterns
\begin{equation}
\tilde\xi^\rho _j=\sum_{\nu=1}^P (\sqrt {\boldsymbol{D}} \boldsymbol{U}^{T})_{\rho\nu}\xi^\nu_j.
\end{equation}
We stress that this rotation has no effect on correlation (recall that $\xi^\mu$ are drawn i.i.d. with equal probability $\mathbb{P}(\xi^\mu_i=\pm1)={1}/{2}$). To check this we can compute the normalized covariance matrix (Bravais-Pearson correlation coefficient):
\begin{equation}\label{eq:Cov}
\text{Cov}(\tilde{{\boldsymbol{\xi}^\rho}}, \tilde{{\boldsymbol{\xi}^\mu}})=\frac{\text{Cov}(\tilde{{\boldsymbol{\xi}^\rho}},\tilde{{\boldsymbol\xi^\mu)}}}{\sqrt{\text{Var}(\tilde{{\boldsymbol\xi^\rho}})\text{Var}(\tilde{{\boldsymbol\xi^\mu}})}}.
\end{equation}
Let us then calculate sample variance and covariance:
\begin{equation}
\text{Var}_S(\tilde{{\boldsymbol\xi^\rho}})=\frac{1}{N}\sum_{i=1}^N( \tilde{{\xi^\rho_i}})^2-\Big(\frac{1}{N}\sum_{i=1}^N \tilde{{\xi^\rho_i}}\Big)^2.
\end{equation}
The first term is easy to compute
\begin{equation}
\frac{1}{N}\sum_{i=1}^N( \tilde{{\xi^\rho_i}})^2=(\sqrt {\boldsymbol{D}} \boldsymbol{U}^{T} \boldsymbol{C}\boldsymbol U \sqrt {\boldsymbol{D}})_{\rho\rho},
\end{equation}
where $C_{\rho\mu}=\frac1N \sum_{i=1}^N \xi^\rho_i \xi^\mu_i$ is the pattern correlation matrix. For large $N$, this tends to the identity matrix, so we have
\begin{equation}
\frac{1}{N}\sum_{i=1}^N( \tilde\xi_i^\rho)^2\underset{N\to \infty}{\sim} D_{\rho\rho}.
\end{equation}
On the other hand, $\sum_i \tilde\xi_i^\rho$ is just the distance covered by a random walk of $N$ steps, so it is $\mathcal{O}(\sqrt{N})$. The average value of the entries of $\tilde {\boldsymbol \xi}^{\rho}$, for large $N$, goes to zero as $\mathcal O (N^{-1/2})$. In conclusion, this proves that
\begin{equation}
\text{Var}_S(\tilde{{\boldsymbol{\xi}^\rho}})\underset{N\to \infty}{\sim} D_{\rho\rho}.
\end{equation}
Regarding the sample covariance matrix, one can see that
\begin{equation}
\text{Cov}_S(\tilde{{\boldsymbol \xi^\rho}},\tilde{{\boldsymbol \xi^\mu}})=\frac1N \sum_{i=1}^N \tilde{\xi}^\rho_i \tilde\xi^\mu_i=\sum_{\nu,\sigma=1}^P(\sqrt {\boldsymbol{D}} \boldsymbol{U}^{T}\boldsymbol{C}  \boldsymbol{U}\sqrt {\boldsymbol{D}})_{\rho \mu}.
\end{equation}
Again, for a large $N$ the correlation matrix $\boldsymbol{C}$ tends to the identity, therefore \eqref{eq:Cov} leads to
\begin{eqnarray}
\text{Cov}_S(\tilde{{\boldsymbol \xi^\rho}},\tilde{{\boldsymbol \xi^\mu}})\underset{N\to \infty}{\sim}\frac{D_{\rho\mu}}{\sqrt{D_{\rho\rho}D_{\mu \mu}}}.
\end{eqnarray}
Since $\boldsymbol{D}$ is diagonal, $D_{\rho\mu}= D_{\mu\mu}\delta_{\rho \mu}$, we get
\begin{equation}
\text{Cov}_S(\tilde{{\boldsymbol \xi^\rho}},\tilde{{\boldsymbol \xi^\mu}})\underset{N\to \infty}{\sim}\frac{D_{\mu\mu}\delta_{\rho \mu}}{\sqrt{D_{\rho\rho}D_{\mu \mu}}}= \delta_{\rho \mu}.
\end{equation}
In the thermodynamic limit, the rotated patterns $\tilde {\boldsymbol \xi}$ are therefore uncorrelated.

\section{Existence of the thermodynamic limit of the pressure}\label{sec:esistenza}
In this Appendix, we report the proof of the existence of the thermodynamic limit for the intensive pressure of the ``relativistic'' Hopfield model with correlated pattern.
To this aim, we consider a system made of $N$ neurons and other two systems made of, respectively, $N_1$ and $N_2$ neurons with $N_1$ and $N_2$ such that
  $N=N_1+N_2$, and we show by interpolation that $F_{N,\beta, a} < F_{N_1, \beta, a} +  F_{N_2, \beta, a}$. 
To this goal we also need to introduce the Mattis magnetizations related to the $P$ patterns in the largest model (made of $N$ neurons) and those in the smaller ones:
\begin{equation*}
\tilde{{m}}_{N_{1}}^\mu=\frac{1}{N_{1}}\sum_{i=1}^{N_{1}}\tilde{{\xi}}^\mu_{i}\sigma
_{i},\quad \tilde{{m}}_{N_{2}}^\mu=\frac{1}{N_{2}}\sum_{i=1}^{N_{2}}\tilde{{\xi}}_{i}^\mu\sigma _{i}.
\end{equation*}%
By denoting the $N_2$ neurons of the second system with $\sigma_{N_1 +1},..., \sigma_{N_1+ N_2}$ we can write 
$$\begin{aligned}
\tilde{{m}}_{N}^\mu&=\frac{1}{N}\sum_{i=1}^{N}\tilde{{\xi}}_{i}^\mu\sigma_{i}=\frac{1}{N}
\Big( \sum_{i=1}^{N_{1}}\tilde{{\xi}}_{i}^\mu\sigma_{i}+\sum_{j=N_1 +1}^{N_1 +N_{2}}
{\tilde{{\xi}}}^\mu_{j}{\sigma}_{j} \Big) =\frac{1}{N}\left( N_{1}\tilde{{m}}^\mu_{N_{1}}+N_{2}\tilde{{m}}^\mu_{N_{2}} \right)
=\rho _{1}\tilde{{m}}_{N_{1}}^\mu+\rho _{2}\tilde{{m}}_{N_{2}}^\mu,
\end{aligned}$$

where $\rho_1$ and $\rho_2$ represents the relative densities
\begin{equation*}
\rho _{1}=\frac{N_{1}}{N}\quad \text{and\quad }\rho _{2}=\frac{N_{2}}{N}.
\end{equation*}
Let us introduce the interpolating parameter $t \in \left[0,1\right]$ used to define the interpolating intensive pressure $F_N(t)$ as follows\footnote{Notice that, in order to enlight the notation, here we dropped the dependence on the set of digital patterns.}
 	\begin{equation}
 	\begin{aligned}
	\label{alpha-intera}
	 F_N \left( t \right) & =  \frac{1}{N} \log \sum_{ \{ \boldsymbol \sigma \} }^{2^{N}}\exp \biggl[ t\beta N\sqrt{1+\tilde{\mathbf{m}}_{N}^{2}} + \left( 1-t\right)\beta \left( N_{1}\sqrt{1+{\tilde{\mathbf{m}}_{N_{1}}}^{2}}+N_{2}\sqrt{1+\tilde{\mathbf{m}}_{N_{2}}^{2}} \right) \biggr],
	\end{aligned}
	\end{equation}
where we omit the subscript $\beta, a$ to lighten the notation.

Notice that in the limit $t \to 1$ and $t \to 0$ we get, respectively,
\begin{equation}
F_N \left( 1\right) =\frac{1}{N}  \log \sum_{\left\{ \boldsymbol \sigma \right\}
}^{2^{N}}\exp \big( \beta N\sqrt{1+\tilde{\mathbf{m}}_{N}^{2}}\big) = F_{N, \beta, a} 
  \label{34}
\end{equation}
and
\begin{align*}
F_N \left(0\right) &=\frac{1}{N}  \log \sum_{\left\{\boldsymbol \sigma
\right\} }^{2^{N}}\exp \left(\beta N_{1}\sqrt{1+\tilde{\mathbf{m}}_{N_{1}}^{2}}+\beta N_{2}
\sqrt{1+\tilde{\mathbf{m}}_{N_{2}}^{2}}\right)  \notag 
 =\frac{1}{N}\mathbb{E}\left( \log
Z_{N_{1},\beta, a}+\log Z_{N_{2},\beta, a}\right)  =\notag \\
&=\rho _{1} F_{N_{1},\beta, a}  +\rho _{2} F_{N_{2},\beta, a}.  \label{35}
\end{align*}%
Applying the fundamental theorem of calculus we have
\begin{equation}
F_N \left( 1\right) = F_N \left( 0\right) +\int_{0}^{1}%
\frac{d F_N \left( t\right) }{d t} \bigg \rvert_{t =t'} dt'. \label{33}
\end{equation}
Now, in order to prove that $F_N \left( \beta \right)$ is sub-additive we just need to prove that the derivative with respect to $t$ of the interpolating pressure is non-positive. By direct evaluation we get

\begin{equation}
\label{30}
\begin{aligned}
\frac{d F_N \left( t\right) }{d t}& = \frac{d
}{d t}\frac{1}{N}\log \sum_{\left \{\boldsymbol  \sigma \right\} }^{2^{N}}\exp
\left [ t\beta N\sqrt{1+\tilde{\mathbf{m}}_{N}^{2}}+\left( 1-t\right)\beta \left ( N_{1}\sqrt{1+%
\tilde{\mathbf{m}}_{N_{1}}^{2}}+N_{2}\sqrt{1+\tilde{\mathbf{m}}_{N_{2}}^{2}}\right) \right]
 \\
&=\frac{1}{N} \big\langle N\sqrt{1+\tilde{\mathbf{m}}_{N}^{2}}-N _{1}\sqrt{1+\tilde{\mathbf{m}}_{N_{1}}^{2}}-N_{2}\sqrt{1+\tilde{\mathbf{m}}_{N_{2}}^{2}}\big\rangle _{t} \\
&=\big\langle \sqrt{1+\tilde{\mathbf{m}}_{N}^{2}}-\rho _{1}\sqrt{1+\tilde{\mathbf{m}}
_{N_{1}}^{2}}-\rho _{2}\sqrt{1+\tilde{\mathbf{m}}_{N_{2}}^{2}}\big\rangle _{t},
\end{aligned}
\end{equation}
where 
%
\begin{equation}
    \label{mediat}
\langle f \rangle_t= \frac{\sum_{\{ \boldsymbol \sigma \}}^{2^N} [e^{- t H_N- (1-t )( H_{N_{1}} + H_{N_{2}} ) }f]}{\tilde{Z}_N(t)},
\end{equation}
is the Boltzmann average with respect to the interpolating partition function $\tilde{Z}_N(t)$ associated to the interpolating pressure \eqref{alpha-intera}, that is
\begin{equation}
\tilde{Z}_N(t) = \sum_{\{ \boldsymbol \sigma \}}^{2^N} e^{- \beta t H_N- \beta(1-t )( H_{N_{1}} + H_{N_{2}} ) }.
\end{equation}
\medskip
\begin{Proposition}\label{proposizioneModelloRelativisticoSubadditivo}  
The $t$-derivative of $F_{N}(t)$ is non positive.
\end{Proposition}

\begin{proof}
Due to the equalities shown above, the proof simply requires that the right hand side of equation (\ref{30}) is smaller than or equal to zero, namely

\begin{equation}
\label{aste}
\sqrt{1+\tilde{\mathbf{m}}_{N}^{2}}-\rho _{1}\sqrt{1+\tilde{\mathbf{m}}_{N_{1}}^{2}}-\rho
_{2}\sqrt{1+\tilde{\mathbf{m}}_{N_{2}}^{2}}\leq 0.
\end{equation}
Since $\rho _{2}=1-\rho _{1}$, Equation (\ref{aste}) can be rewritten in terms of the variables $\tilde{\mathbf{m}}_{N_{1}}$, $\tilde{\mathbf{m}}_{N_{2}}$ and $\rho _{1}$ only (in the following, in order to lighten the notiation we will denote $\rho_1$ with $\rho$, and, in particular, $ \tilde{\mathbf{m}}_{N}=\rho
\tilde{\mathbf{m}}_{N_{1}}+\left( 1-\rho \right) \tilde{\mathbf{m}}_{N_{2}}$). We rewrite equation (\ref{aste}) as 
\begin{equation}
\sqrt{1+\left( \rho \tilde{\mathbf{m}}_{N_{1}}+\left( 1-\rho \right) \tilde{\mathbf{m}}
_{N_{2}}\right) ^{2}}-\rho \sqrt{1+\tilde{\mathbf{m}}_{N_{1}}^{2}}-\left( 1-\rho
\right) \sqrt{1+\tilde{\mathbf{m}}_{N_{2}}^{2}} \leq 0. 
\end{equation}
In order to prove the inequality it is sufficient to notice that $$f:x\mapsto \sqrt{1+x^{2}},$$ is a convex function, that is, for $ \lambda \in \left[ 0,1\right] $, we have

\begin{equation*}
f\left( \lambda x_{1}+\left( 1-\lambda \right) x_{2}\right) \leq \lambda
f\left( x_{1}\right) +\left( 1-\lambda \right) f\left( x_{2}\right).
\end{equation*}%
Now, by identifying $\lambda =\rho,  x_{1}=\mathbf{m}_{N_{1}},  x_{2}=\mathbf{m}_{N_{2}},$ 
since $\mathbf{m}_{N}=\rho \mathbf{m}_{N_{1}}+ (1-\rho) \mathbf{m}_{N_{2}}$, we see that
\begin{equation*} 
\sqrt{1+\mathbf{m}_{N}^{2}}=\sqrt{1+\left( \rho _{1}\mathbf{m}_{N_{1}}+\rho
_{2}\mathbf{m}_{N_{2}}\right) ^{2}}\leq \rho _{1}\sqrt{1+\mathbf{m}%
_{N_{1}}^{2}}+\rho _{2}\sqrt{1+\mathbf{m}_{N_{2}}^{2}},
\end{equation*}
which proves our assertion.
\end{proof}

The proposition
\ref{proposizioneModelloRelativisticoSubadditivo} allows us to state that
\begin{equation*}
F_N \left( 1\right) - F_N \left( 0\right) =\int_{0}^{1}%
\frac{d F_N \left( t \right) }{d t} \bigg \rvert_{t=t'} dt' \leq 0,
\end{equation*}
meaning that
\begin{equation}
N F_{N, \beta, a } \leq N _{1} F _{N_{1}, \beta, a} + N_{2} F_{N_{2}, \beta, a}.  \label{36}
\end{equation}
By means of the Fekete's lemma, this result proves {\it de facto} the following\medskip
\begin{Theorem}
The thermodynamic limit of the intensive pressure of the relativistic model with temporally correlated patterns (Def. \eqref{Hopfield-relativistico}) exists and equals the lower bound of the sequence $F_{N, \beta, a}$ :
\begin{equation*}
\exists \lim_{N\rightarrow \infty } F_{N, \beta, a }
=\inf_{N \in \mathbb{N}} \left\{ F_{N, \beta, a} \right\}
= F_{ \beta, a}.
\end{equation*}
\end{Theorem}

 \section{Critical curve of ergodicity breaking phase transition} \label{sec:Tc_line}
This section aims to analytically determine the transition line dividing the ergodic phase from the spin-frosted one. Since we expect (as we numerically checked) that the ergodicity breaking takes place with a continuous phase transition in the order parameters, it is natural to study the self-consistency equations near the point $\boldsymbol M=0$. To do this, we remind that the self-consistency equation takes the following vectorial form
\begin{equation}
\label{vSC}
M_{\mu}  =\frac{(1+\boldsymbol{M}^T\boldsymbol{X} \boldsymbol{M})\mathbb E \,{\xi}^{\mu} \tanh\Big( \beta \frac{\sum_\rho {\xi}^\rho  (\boldsymbol{X} \boldsymbol{ M})_\rho}{\sqrt{1+\boldsymbol{M}^T\boldsymbol{X} \boldsymbol{M}}} \Big)}{1+\sum_\nu (\boldsymbol M^T \boldsymbol X)_{\nu }\mathbb E \,{\xi}^{\nu} \tanh\Big( \beta \frac{\sum_\rho {\xi}^\rho  (\boldsymbol{X} \boldsymbol{ M})_\rho}{\sqrt{1+\boldsymbol{M}^T\boldsymbol{X} \boldsymbol{M}}} \Big)} , \quad \ \forall \mu=1,...,P.
\end{equation}
For our concerns, we only need to expand the r.h.s. of the self-consistency equations up to the linear order in $\boldsymbol M$. In this case, it is easy to note that, for $\boldsymbol M\ll 1$, we have
\begin{equation}
\sqrt{1+\boldsymbol{M}^T\boldsymbol{X} \boldsymbol{M}}\simeq 1+\frac12 \boldsymbol{M}^T\boldsymbol{X} \boldsymbol{M},
\end{equation}
thus we see that non-trivial contributions from the relativistic expansion start from the second order in $\boldsymbol M$. Then, in this limit the argument of the hyperbolic tangent for the relativistic model reduces to its classical counterpart:
\begin{equation}
\tanh\Big( \beta \frac{\sum_\rho {\xi}^\rho  (\boldsymbol{X} \boldsymbol{ M})_\rho}{\sqrt{1+\boldsymbol{M}^T\boldsymbol{X} \boldsymbol{M}}} \Big)\simeq \tanh\Big( \beta {\sum_\rho {\xi}^\rho  (\boldsymbol{X} \boldsymbol{ M})_\rho} \Big)= \beta  {\sum_\rho {\xi}^\rho  (\boldsymbol{X} \boldsymbol{ M})_\rho}+\mathcal O (\boldsymbol M^2).
\end{equation}
Furthermore, both the factor $1+\boldsymbol{M}^T\boldsymbol{X} \boldsymbol{M}$ and the denominator in \eqref{vSC} have non-trivial contribution only at the order $\mathcal O (\boldsymbol M^2)$, thus we arrive to
\begin{equation}
M_\mu =\beta \sum_\rho \mathbb E\ \xi^\mu  \xi^\rho (\boldsymbol X \boldsymbol M)_\rho. 
\end{equation}
Now, the average $\mathbb E$ only apply on the product $\xi^\mu \xi ^\rho$. In particular, we see that $\mathbb E \xi^\rho \xi^\mu$ is the (theoretical) covariance of the $\xi$ variables, which turns out to be trivially $\mathbb E \xi^\rho \xi^\mu=\delta_{\mu\rho}$. With this results and by making explicit the structure of the temporal correlation matrix $\boldsymbol X$, we arrive at the result
\begin{equation}
\label{eq:nearErg}
M_\mu =\beta(M_\mu +a M_{\mu +1}+ aM_{\mu -1}).
\end{equation}
Now, we can further note that, by starting from the ergodic region and lowering the thermal noise $T=\beta^{-1}$, we enter in a thermodynamic phase which is characterized by fully symmetric states of the form $\boldsymbol{M}^T=(1,1,...1)m$ with $m\ne 0$. Then, by replacing this ansatz in Equation \eqref{eq:nearErg}, we get $m=\beta(1+2a)m$, from which we immediately get $\beta_c=\frac{1}{1+2a}$. Thus, the critical curve for the ergodicity breaking in the plane $(T,a)$ is given by the equation
\begin{equation}
T_c(a)=1+2a.
\end{equation}
Thus, if $T>T_c$ the only solution is $\boldsymbol{M}=\boldsymbol{0}$ while, if $T<T_c$ there exist solutions $\boldsymbol{M}\ne\boldsymbol{0}$. The line $T_c$ analytically found is consistent with the numerical solution of the equation \eqref{vSC}. 

%

\section*{Acknowledgments}
The Authors acknowledge partial financial fundings by Universit\`a Sapienza di Roma (Progetto Ateneo RG11715C7CC31E3D) and by ``Rete Match - Progetto Pythagoras''  (CUP:J48C17000250006).

\section*{Data availability}
Data sharing is not applicable to this article as no new data were created or analyzed in this study.

\bibliographystyle{unsrt} 

\end{document}